\documentclass[11pt]{article}
\usepackage{authblk}
\usepackage{amsfonts}
\usepackage{fancyhdr}
\usepackage[a4paper, top=2.5cm, bottom=2.5cm, left=2.2cm, right=2.2cm]{geometry}
\usepackage{cite}
\usepackage[numbers,sort]{natbib}
\usepackage{changepage}
\usepackage[shortlabels]{enumitem}
\usepackage{soul}
\usepackage{asymptote}
\usepackage{amssymb}
\usepackage{bbm}
\usepackage{graphicx}
\usepackage{hyperref}
\usepackage{float}
\hypersetup{colorlinks=true, citecolor=blue, linkcolor=black}
\usepackage{amsfonts}
\usepackage{amsmath}
\usepackage{amssymb}
\usepackage{amsthm}
\usepackage{thm-restate}
\usepackage{enumitem} 
\usepackage{graphicx} 
\usepackage{algorithm,algorithmicx}
\usepackage[noend]{algpseudocode}
\usepackage[skins]{tcolorbox}
\usepackage{nicefrac}
\usepackage{subcaption}
\usepackage{cleveref}
\usepackage{xspace}
\usepackage[T1]{fontenc}
\usepackage[framemethod=tikz]{mdframed}

\newtheorem{theorem}{Theorem}[section]
\newtheorem{corollary}[theorem]{Corollary}
\newtheorem{lemma}[theorem]{Lemma}

\newtheorem{proposition}[theorem]{Proposition}
\newtheorem{claim}[theorem]{Claim}
\newtheorem{observation}[theorem]{Observation}

\newcommand{\calI}{\mathcal{I}}
\newcommand{\calM}{\mathcal{M}}
\newcommand{\calA}{\mathcal{A}}
\newcommand{\calF}{\mathcal{F}}
\newcommand{\calD}{\mathcal{D}}

\newcommand{\Cov}{\mathrm{Cov}}

\newcommand{\EE}{\mathbb E}
\newcommand{\NN}{\mathbb N}
\newcommand{\PP}{\Pr}

\newcommand{\ctwo}{\frac{\nicefrac{1}{2} - c}{\nicefrac{1}{2} + c}}
\newcommand{\finalCconstant}{0.01}

\theoremstyle{definition} \newtheorem{definition}[theorem]{Definition}
\newcommand{\tristannote}[1]{\textcolor{blue}{Tristan: #1}}
\newcommand{\davidnote}[1]{\textcolor{purple}{David: #1}}

\newcommand{\opton}{OPT_{on}}

\usepackage{todonotes}
\newcounter{todocounter}
\newcommand{\todonum}[2][]{\stepcounter{todocounter}\todo[#1]{\thetodocounter: #2}}

\newcommand{\notshow}[1]{}
\newcommand{\ride}{\ensuremath{\textsc{RideHail}}\xspace}

\newenvironment{wrapper}[1]
{
	\smallskip
	\begin{center}
		\begin{minipage}{\linewidth}
			\begin{mdframed}[hidealllines=true, backgroundcolor=gray!20, leftmargin=0cm,innerleftmargin=0.4cm,innerrightmargin=0.4cm,innertopmargin=0.4cm,innerbottommargin=0.4cm,roundcorner=10pt]
				#1}
			{\end{mdframed}
		\end{minipage}
	\end{center}
	\smallskip
}

\usepackage{etoolbox}
\usepackage{tikz}
\usetikzlibrary{tikzmark}
\usetikzlibrary{calc}

\errorcontextlines\maxdimen

\newcommand{\ALGtikzmarkcolor}{black}
\newcommand{\ALGtikzmarkextraindent}{4pt}
\newcommand{\ALGtikzmarkverticaloffsetstart}{-.5ex}
\newcommand{\ALGtikzmarkverticaloffsetend}{-.5ex}
\makeatletter
\newcounter{ALG@tikzmark@tempcnta}
\newcommand\ALG@tikzmark@start{%
	\global\let\ALG@tikzmark@last\ALG@tikzmark@starttext%
	\expandafter\edef\csname ALG@tikzmark@\theALG@nested\endcsname{\theALG@tikzmark@tempcnta}%
	\tikzmark{ALG@tikzmark@start@\csname ALG@tikzmark@\theALG@nested\endcsname}%
	\addtocounter{ALG@tikzmark@tempcnta}{1}%
}

\def\ALG@tikzmark@starttext{start}
\newcommand\ALG@tikzmark@end{%
	\ifx\ALG@tikzmark@last\ALG@tikzmark@starttext
	\else
	\tikzmark{ALG@tikzmark@end@\csname ALG@tikzmark@\theALG@nested\endcsname}%
	\tikz[overlay,remember picture] \draw[\ALGtikzmarkcolor] let \p{S}=($(pic cs:ALG@tikzmark@start@\csname ALG@tikzmark@\theALG@nested\endcsname)+(\ALGtikzmarkextraindent,\ALGtikzmarkverticaloffsetstart)$), \p{E}=($(pic cs:ALG@tikzmark@end@\csname ALG@tikzmark@\theALG@nested\endcsname)+(\ALGtikzmarkextraindent,\ALGtikzmarkverticaloffsetend)$) in (\x{S},\y{S})--(\x{S},\y{E});%
	\fi
	\gdef\ALG@tikzmark@last{end}%
}

\apptocmd{\ALG@beginblock}{\ALG@tikzmark@start}{}{\errmessage{failed to patch}}
\pretocmd{\ALG@endblock}{\ALG@tikzmark@end}{}{\errmessage{failed to patch}}
\makeatother

\algblock[with]{With}{EndWith}
\algblockdefx[With]{With}{EndWith}%
[1]{\textbf{with} #1 \textbf{do}}%
{}
\makeatletter
\ifthenelse{\equal{\ALG@noend}{t}}%
{\algtext*{EndWith}}
{}%
\makeatother

\usepackage{asyfig} 

\begin{document}
\title{Online Stochastic Max-Weight Bipartite Matching: \\
Beyond Prophet Inequalities\footnote{Research supported in part by NSF Awards CCF1763970, CCF191070, CCF1812919, ONR award N000141912550, and a gift from Cisco Research.}}

\author[1]{Christos Papadimitriou}
\author[2]{Tristan Pollner}
\author[2]{Amin Saberi}
\author[2]{David Wajc}

\affil[1]{Columbia University}
\affil[2]{Stanford University}
\date{\vspace{-1.5cm}}

\pagenumbering{gobble}

\maketitle
\begin{abstract} 
The rich literature on online Bayesian selection problems has long focused on so-called prophet inequalities, which compare the gain of an online algorithm to that of a ``prophet'' who knows the future. An equally-natural, though significantly less well-studied benchmark is the optimum \emph{online} algorithm, which may be omnipotent (i.e., computationally-unbounded), but not omniscient. What is the computational complexity of the optimum online? How well can a polynomial-time algorithm approximate it?

\smallskip 

We study the above questions for the online stochastic maximum-weight matching problem under vertex arrivals. For this problem, a number of  $\nicefrac{1}{2}$-competitive algorithms are known against optimum offline. This is the best possible ratio for this problem, as it generalizes the original single-item prophet inequality problem. 

\smallskip 

We present a polynomial-time algorithm which approximates the optimal \emph{online} algorithm within a factor of $0.51$---beating the best-possible prophet inequality. 
In contrast, we show that it is \textsc{PSPACE}-hard to approximate this problem within some constant $\alpha < 1$. 
 \end{abstract}

\newpage
\pagenumbering{arabic}

\section{Introduction}

Decision-making in an uncertain, dynamic environment influenced by one's decisions has arguably
always been the essence of life, and yet it appears to have been first confronted mathematically 
by Herbert Robbins and Richard Bellman, from different perspectives, in the late 1940s and early 1950s.  
Decision theory initially focused on instantaneous decisions, but later gave us stopping rules
and the gem of prophet inequalities \cite{krengel1978semiamarts}.  Later, the Internet age brought us new business models 
relying exclusively on stochastic decision making --- online advertising, ride hailing,
kidney exchanges --- in which the changing environment affected by 
the agents' decisions can often be abstracted as an evolving weighted bipartite graph.  

Here we study one such problem,
the online Bayesian bipartite matching, or $\ride$, problem. 
The input to this problem is a random bipartite graph, revealed over time.
Initially, the $m$ nodes on one side of the graph, termed taxis or bins, are present.
The $n$ nodes on the other side, termed passengers or balls, are revealed over time, in a fixed order known to us.
Initially, we know for each ball $t$ the probability $p_t$ of it actually arriving, as well
as the non-negative weight $w_{i,t}$ of the edge connecting it to any bin $i$ --- if it arrives.
If ball $t$ does not arrive, we do nothing at time $t$; if it does arrive, we can choose to match it, irrevocably, to some unmatched neighbor $i$ before time $t+1$,  yielding a profit $w_{i,t}$.
Our goal is to maximize the overall expected profit.

$\ride$ generalizes the classic single-item online Bayesian stopping rule problem --- the so-called \emph{prophet inequality} problem.
In particular, our problem with a single offline node already captures the worst-case instances of the prophet inequality, for which no online algorithm is better than $\nicefrac{1}{2}$-competitive against the optimal offline algorithm. 
On the other hand, $\ride$ is a special case of the unit-demand combinatorial auctions problem and online stochastic maximum-weight matching, for both of which $\nicefrac{1}{2}$-competitive algorithms are known \cite{feldman2015combinatorial, dutting2020prophet,ezra2020online}.

There is an extensive literature on numerous variations of online Bayesian selection problems, relating the performance of online algorithms with the omniscient prophet of inequality fame --- that is to say, with the offline optimum (see \Cref{sec:related-work}). 
In particular, these works study achievable competitive ratios: the worst-case ratio over all inputs between the online algorithm and the best offline algorithm.
While this may be the right thing to do when the input is adversarial, when the input is generated {\em stochastically} one perhaps could do better.
In particular, in the stochastic case, the optimum online algorithm {\em for the given input} is a well-defined benchmark that can be computed in exponential time. Suddenly we are in the realm of {\em approximation algorithms,} rather than of competitive analysis.

In approximation algorithms, typically one explores two interesting questions: First, is approximation hard?  And second, what is the best approximation ratio achievable in polynomial time? 

In this paper we address both questions. First, we show that for some $\alpha<1$ it is \textsc{PSPACE}-hard to approximate the $\ride$ problem within a factor of $\alpha$.
\begin{wrapper}
\begin{restatable}{theorem}{pspacehard}\label{thm:ridehailpspacehard}
    It is \textsc{PSPACE}-hard to approximate the optimal online $\ride$ algorithm within a factor of $\alpha$, for some absolute constant $\alpha<1$. This remains true even when all weights and inverse arrival probabilities are bounded by some polynomial in the size of the input.
\end{restatable}
\end{wrapper}

Here, $1-\alpha$ is small, limited by the current status of expander constructions and approximation hardness of MAX-SSAT (see \Cref{sec:prelims}). To our knowledge, no past work on variants of online matching had demonstrated such level of hardness. 
(We briefly note that PSPACE is the ``right'' complexity class for this problem, which can be solved in polynomial space by standard techniques.)
Finally, we note (see  \Cref{app:policyhardnesschernoff}) that our hardness of approximation result directly implies hardness of computing an approximately-optimal online algorithm, and not just its expected value.\footnote{Note that for some problems, although computing the expected value of the optimal policy is hard, computing the optimal policy itself is actually easy. For example, computing the probability that a random graph containing each edge $e$ with probability $p_e$ is connected is the
\#P-hard the network reliability problem \cite{karger2001randomized,valiant1979complexity,provan1983complexity}, while computing connectivity in the realized graph (even in an online setting) is trivial.} 


We then develop an approximation algorithm, and a technique to bound the (online) optimum. To our knowledge, all past work on approximating the large family of online Bayesian selection problems, with the exception of \cite{anari2019nearly}, has used the prophet inequality benchmark of the offline optimum, which necessarily limits the approximation ratio for many variations to be at most $\nicefrac{1}{2}$.

{\em We go for bounding the online optimum.}  We achieve this by identifying a new constraint which separates online from offline algorithms. In particular, we note that online algorithms cannot match an edge $(i,t)$ with probability greater than the probability of ball $t$ arriving, times the probability of bin $i$ not being matched by the online algorithm beforehand, due to the independence of these events. 
This constraint, which is not true of offline algorithms, poses restrictions on the marginal probabilities of edges to be matched by the optimal online algorithm.
Combining this constraint with the natural matching constraints we obtain a new LP which bounds the optimal online algorithm's gain.
Using this new LP bound (and a number of further ideas, see \Cref{sec:techniques}), we design a new algorithm which recovers at least $51\%$ of the online optimum, i.e., a
ratio strictly better than the optimal competitive ratio of $\nicefrac{1}{2}$.

\begin{wrapper}
\begin{restatable}{theorem}{thmalg}\label{thm:alg}
    There exists a polynomial-time online algorithm which is a $0.51$-approximation of the optimal online algorithm for the  $\ride$ problem.
\end{restatable}
\end{wrapper}

We further generalize our algorithm and achieve the same approximation bound for the more general problem in which weights of any given ball's edges can follow any joint distribution, but weights of different ball's edges are independent. That is, we extend our positive results to the more general bipartite weighted matching problem studied by prior work~\cite{ezra2020online,dutting2020prophet,feldman2015combinatorial}. (See \Cref{sec:general-algo}.)



\subsection{Techniques}\label{sec:techniques}
Here we give a very brief overview of the key ideas used to obtain our main results.

\subsubsection{Hardness} For our \textsc{PSPACE}-hardness result, we first refine the result of \citet{condon1997random} for maximum satisfiability of stochastic SAT instances.
In the stochastic SAT (SSAT) problem, introduced by \citet{papadimitriou1985games},
a 3CNF formula is given, and variables $x_1,x_2,\dots,x_n$ are alternatingly set by an (online) algorithm and randomly set by nature.
\citet{condon1997random} proved that approximating the maximum  expected number of satisfied clauses of an SSAT instance is \textsc{PSPACE}-hard. 
Using an \emph{expander graph} construction, we extend this result to SSAT instances in which each variable appears in at most $O(1)$ clauses. 
We then give a polynomial-time reduction from approximating maximum 
satisfiability of a bounded-occurrence SSAT instance to approximating the optimal online algorithm for the \ride problem, implying our claimed \textsc{PSPACE}-hardness.

\subsubsection{Algorithm} Our algorithmic results involve a number of ideas. We outline the key ones here.
\paragraph{\underline{Our LP Benchmark.}} 
We recall that we want to approximate the optimal online algorithm within a factor strictly greater than $\nicefrac{1}{2}$ (which is tight against the optimal offline algorithm). Hence, our first objective is to identify a property which separates online from offline algorithms.
To this end, we note (as did \cite{torrico2017dynamic}) that for any online algorithm $\calA$, the event of the arrival of ball $t$ is independent of the event that bin $i$ is not matched by Algorithm $\calA$ prior to time $t$. (Note that this constraint does not necessarily hold for the prophetic optimum offline algorithm, which makes its matching choices based on both past and future balls' arrivals.)
Consequently, the probability that  edge $(i,t)$ is matched by online algorithm $\calA$ is at most the product of these two events' probabilities.
Combining this constraint with natural matching constraints, we obtain an LP which bounds the expected gain of the optimal online algorithm (but not its offline counterpart).
In \Cref{sec:LP-Match-obs} we note that this LP completely characterizes the optimal online algorithm for instances with a single offline node, equivalent to the single-item online Bayesian selection problem. This is not true for general instances (as we would expect due to \Cref{thm:ridehailpspacehard}); we therefore use this LP to \emph{approximate} the optimal online policy.

\paragraph{\underline{A Second Chance Algorithm.}}
We present an efficient online algorithm for approximately rounding a solution to the above LP. 
Let $x_{i,t}$ be the decision variables of this LP.
Intuitively, these $x_{i,t}$ serve as proxies for the probability of $(i,t)$ to be matched by the optimal online algorithm. 
Our online algorithm matches each edge $(i,t)$ with probability at least $x_{i,t}\cdot (\nicefrac{1}{2}+c)$ for $c=\nicefrac{1}{100}$.
Our algorithm can be seen as a generalization and extension of the $\nicefrac{1}{2}$-competitive algorithm of Ezra et al.~\cite{ezra2020online} for our problem.
Their algorithm can be thought of as approximately rounding the above LP (without the new constraint) as follows. After each arrival of ball $t$, pick a bin $i$ with probability proportional to $x_{i,t}$, and then, if bin $i$ is unmatched, match edge $(i,t)$ with some probability $q_{i,t}$. These $q_{i,t}$ are set to guarantee that each edge $(i,t)$ is matched with marginal probability $x_{i,t}\cdot \nicefrac{1}{2}$, which can be thought of as applying an online contention resolution scheme as in \cite{feldman2016online}.
To improve on this, 
we first note that modifying these $q_{i,t}$ appropriately results in each edge $(i,t)$ being matched with probability precisely $x_{i,t}\cdot (\nicefrac{1}{2}+c)$ if $\sum_{t'<t} x_{i,t'}$ is small, and 
at least $x_{i,t}\cdot (\nicefrac{1}{2}-O(c))$ otherwise.
To increase these marginal probabilities to $x_{i,t}\cdot (\nicefrac{1}{2}+c)$ for each edge $(i,t)$, we repeat the above process if $t$ is unmatched, letting $t$ pick a second bin $i'$ and possibly matching edge $(i',t)$. For this second pick to achieve its desired effect, bin $i$ should not be matched too often when picked by ball $t$ in its second pick.
That is, conditioning on $t$ not being matched after its first pick should not decrease the probability of $i$ being free by too much. This is the core of our analysis.

\paragraph{\underline{Analysis.}}
To prove that conditioning on ball $t$ not being matched after its first pick indeed does not decrease the  probability of bin $i$ being free by much, we show that (i) the bins' matched statuses by time $t$ have low correlation, and (ii) bin $i$ is unlikely to be picked twice by ball $t$.
To prove Property (i), we show that most of the probability of a bin to be matched by this algorithm is accounted for by variables which are negatively correlated, and even \emph{negatively associated} (see \Cref{sec:prelims}). For our proof of Property (ii), we finally reap the rewards from our new LP constraint. In particular, this constraint implies that for bins $i$ with $\sum_{t'<t} x_{i,t'}$ large, as above, $x_{i,t}$ must be low, implying that bin $i$ is unlikely to be picked by ball $t$ as its first pick.
Properties (i) and (ii) together imply that conditioning on ball $t$ not being matched after its first pick does not decrease the probability of bin $i$ to be unmatched much. This then implies that the second pick is not too unlikely to result in a match of edge $(i,t)$. We thus find that each edge $(i,t)$ is matched by our algorithm with probability $(\nicefrac{1}{2}+c)\cdot x_{i,t}$, from which our $(\nicefrac{1}{2}+c)$-approximation follows.

\subsection{Related Work}\label{sec:related-work}
The literature on online Bayesian selection problems is a long and illustrious one. We briefly outline some of the most relevant work here. See also surveys on the topic \cite{correa2019recent,hill1992survey,lucier2017economic,hartline2012approximation}.

A seminal result in the stopping theory literature, the first prophet inequality, a $\nicefrac{1}{2}$-competitive algorithm for the single-item online Bayesian selection problem, was first given in the late 70s \cite{krengel1978semiamarts}. Multiple algorithms achieving this bound are known \cite{samuel1984comparison,alaei2014bayesian,kleinberg2019matroid,ezra2020online,alaei2012online}.
On the other hand, better bounds are known for various special cases, most prominently for i.i.d.~distributions \cite{correa2017posted,hill1982comparisons,abolhassani2017beating}.

Numerous \emph{multiple-item} online Bayesian selection problems were studied over the years. 
Generalizations of the classic $\nicefrac{1}{2}$-competitive prophet inequality of \cite{krengel1978semiamarts} for single-items were given for matroid constraints \cite{kleinberg2019matroid}, 
for multiple items \cite{alaei2014bayesian},
for bipartite matching under one-sided vertex arrivals \cite{alaei2012online,feldman2015combinatorial},
and for general matching under vertex arrivals \cite{ezra2020online}, with positive results known for many other constraints \cite{dutting2020prophet,feldman2016online,feldman2015combinatorial,kleinberg2019matroid, dutting2020log}.
For matching under \emph{edge} arrivals, a number of positive results are known \cite{feldman2016online,kleinberg2019matroid,gravin2019prophet}, and a competitive ratio of $\nicefrac{1}{2}$ is impossible for this stochastic problem  \cite{gravin2019prophet,ezra2020online}. This mirrors a similar separation between vertex arrivals and edge arrivals for this problem's (unweighted) deterministic counterpart \cite{gamlath2019online}.
Much of this work on approximating the optimal offline algorithm (prophet inequalities) for online Bayesian selection problems was motivated by connections discovered between prophet inequalities and algorithmic mechanism design \cite{hajiaghayi2007automated,chawla2010multi,correa2019pricing,feldman2015combinatorial}.
The computational complexity of approximating the \emph{online} optimal algorithm, however, was significantly less well studied.

The only previous positive result for approximating the online optimum algorithm (better than offline optimum) for an online Bayesian selection problem is due to Anari et al.~\cite{anari2019nearly}, who gave a PTAS for a special class of matroid constraints.
On the computational complexity front, the only hardness for such problems we are aware of is the recent result of Agrawal et al.~\cite{agrawal2020optimal}, who show that computing the optimal \emph{ordering} of the random variables for a single-item problem is \textsc{NP}-hard (with an EPTAS for this ordering problem due to \cite{segev2020efficient}).
The (in)approximability of the  optimum online algorithm was studied for other stochastic online optimization problems recently, including probing problems \cite{goel2010probe, chen2016combinatorial, fu2018ptas,segev2020efficient},
stochastic matching problems in infinite-horizon settings under Poisson arrivals and departures \cite{aouad2020dynamic}, two-stage stochastic matching problems \cite{feng2021two}, and stochastic dynamic programming problems \cite{fu2018ptas}.
The computational complexity of approximating $\opton$ for these and other problems remains an intriguing open problem.
We are hopeful that the tools we develop here will prove useful in extending the literature on computational complexity and approximability of such problems of decision-making under uncertainty.

\paragraph{\textbf{Follow-up work:}} Following this work, the last two authors have extended this paper's algorithm to obtain improved algorithms for the (seemingly unrelated) online edge coloring problem \cite{saberi2021greedy}. Their ideas can be used to improve our approximation ratio from $0.51$ to $0.526$. In another work, Kessel et al.~\cite{kessel2021stationary} study a stationary version of the prophet inequality problem, and obtain optimal competitive ratios, and improved approximation of the optimal online algorithm. Whether other online Bayesian selection problems admit better (efficient) approximation of their optimal online algorithms compared to the optimal prophet inequality remains to be seen. (See \Cref{sec:conclusion}.)

\section{Preliminaries}\label{sec:prelims}

For any algorithm $\mathcal{A}$ and instance $\calI$ of a problem $\Pi$, we let $\mathcal{A}(\calI)$ denote the value of the output of algorithm $\mathcal{A}$ on instance $I$. We use $\opton^{\Pi}(\calI)$ to denote an optimal online algorithm for $\Pi$ on $\calI$. Since the problem $\Pi$ will be clear from context, we will usually just write $\opton(\calI)$.
Our interest is in understanding how well this value can be approximated by efficient online algorithms.
Throughout, we say an algorithm gives an $\alpha$-approximation to a quantity $Q$, for $\alpha\in (0,1)$, if it outputs a number in the range $[\alpha Q, Q]$. 
The following simple fact, whose proof is deferred to \Cref{app:ommittedproofs}, is useful for reductions involving hardness of approximation. 

\begin{restatable}{fact}{approxfact}\label{lem:approxlem}
Let $Q, Q' \geq 0$ be positive quantities, such that $Q'/Q\leq \beta$, and let $\alpha \in (0,1)$. Then, an $\big(\frac{\alpha+\beta}{1 + \beta }\big)$-approximation to $Q + Q'$ yields an $\alpha$-approximation to $Q$.
\end{restatable} 

We now turn to providing background on problems and tools used in this work.

\paragraph{\underline{Stochastic SAT.}}\label{sec:SSAT}

The stochastic SAT (SSAT) problem was first defined by Papadimtriou \cite{papadimitriou1985games}. 
In this work, we will consider the maximization variant of this problem, defined below. 

\begin{definition} The input to the MAX-SSAT problem is a 3CNF formula $\phi$ over an ordered list of variables $(x_1, x_2, \ldots, x_n)$. We choose a value of either \texttt{True} or \texttt{False} for $x_1$, nature chooses a value of either \texttt{True} or \texttt{False} for $x_2$ uniformly at random, we choose a value of either \texttt{True} or \texttt{False} for $x_3$, and so on. Our goal is to maximize the expected number of satisfied clauses in $\phi$ after all the variables have been assigned a value. We will refer to $\{x_1, x_3, \ldots \}$ as the ``deterministic variables'' and $\{x_2, x_4, \ldots \}$ as the ``random variables.''
\end{definition}



In his work introducing SSAT, \citet{papadimitriou1985games}  proved PSPACE-hardness of determining the probability of satisfiability of an SSAT instance. 
Over a decade later, this was improved to a 
\emph{hardness of approximation} result by \citet{condon1997random}, via extensions of the PCP theorem \cite{arora1998proof}. 
In particular, they prove the following hardness of approximation result.

\begin{restatable}{lemma}{randomdebaters}(\cite[Theorem 3.3]{condon1997random}) \label{thm:k-max-ssat-hard}
There exist constants $k \in \NN$ and $\alpha \in (0,1)$ so that it is PSPACE-hard to compute an $\alpha$-approximation to $\opton(\phi)$ for a MAX-SSAT instance $\phi$ satisfying:
\begin{enumerate}
    \item no random variable  appears negated in any clause of $\phi$, and
    \item each random variables appears in at most $k$ clauses of $\phi$.
\end{enumerate}
\end{restatable}

It is worth noting that Theorem 3.3 in \cite{condon1997random} only includes the statement about random variables being non-negated. The second property is a direct consequence of the proof of the theorem. In \Cref{app:ommittedproofs} we explain the necessary modifications to the proof to add this guarantee.

\paragraph{\underline{Expander Graphs.}}


Define the expansion of a graph $G$ as 
$$h(G) := \min_{S \subseteq V, |S| \le |V|/2} \frac{|E(S, V \setminus S)|}{|S|},$$
where $E(X,Y):=\{(x,y)\in E\mid e\in X, y\in Y\}$ denotes the edges with one endpoint in $X$ and the other in $Y$. 
We will utilize results providing explicit, deterministic constructions  of graphs with constant degree and constant expansion (e.g. \cite{gabber1981explicit,lubotzky1988ramanujan}). 

\begin{lemma} \label{lem:explicitexpanders}
There exists a deterministic, polynomial-time construction of a graph on $n$ vertices with expansion at least 1 and maximum degree at most some constant $d$.
\end{lemma}

\paragraph{\underline{Negative Association.}}
\label{sec:prelimNA} 

We briefly review some notions of negative dependence we need in this work, in particular, the notion of \emph{Negatively Associated} random variables.


\begin{definition}[\cite{khursheed1981positive,joag1983negative}]\label{def:NA}
	Random variables $X_1,\dots,X_n$ are \emph{negatively associated (NA)}, if every two monotone non-decreasing functions $f$ and $g$ defined on disjoint subsets of the variables in $\vec{X}$ are negatively correlated. That is,
	\begin{equation}\label{eq:NA}
	\EE[f\cdot g] \leq \EE[f]\cdot \EE[g].
	\end{equation}
\end{definition}

A family of independent random variables are trivially negatively associated. A more interesting example of negatively associated random variables is the following.

\begin{proposition}[0-1 Principle \cite{dubhashi1996balls}]\label{0-1-NA}
	Let $X_1,\dots,X_n\in \{0,1\}$ be binary random variables such that $\sum_i X_i\leq 1$ always. Then, the joint distribution $(X_1,\dots,X_n)$ is negatively associated.
\end{proposition}


More elaborate NA distributions can be obtained via the following closure properties.
\begin{proposition}[NA Closure Properties \cite{khursheed1981positive,joag1983negative,dubhashi1996balls}]\label{NA-closure}
	$\phantom{a}$
	\begin{enumerate}
		\item \label{P7_union} \underline{Independent union.}
		Let $(X_1,\dots,X_n)$ and $(Y_1,\dots,Y_m)$ be 
		two mutually independent negatively associated joint distributions. Then, the joint distribution  $(X_1,\dots,X_n,Y_1,\dots,Y_m)$ is also NA.
		\item \label{P6_inc_funs} \underline{Function composition.}
		Let\, $\mathbf{X}= (X_1,\dots,X_n)$ be  NA, and let $f_1,\dots,f_k$ be monotone non-decreasing functions defined on disjoint subsets of\, $\mathbf{X}$. Then the joint distribution  $(f_1,\dots,f_k)$ is also NA.
	\end{enumerate}
\end{proposition}

%




Negative association implies many powerful concentration inequalities and other useful properties (see e.g., \cite{dubhashi1996balls,khursheed1981positive,joag1983negative,asadpour2017log}). For our purposes we will use the pairwise negative correlation of NA variables, implied by \Cref{eq:NA} with the disjoint functions $f(\vec{X})=X_i$ and $g(\vec{X})=X_j$ for $i\neq j$.

\begin{proposition}\label{NA:neg-corr}
	Let $X_1,\dots,X_n$ be NA random variables. Then, for all $i\neq j$, 
	$\Cov(X_i,X_j)\leq 0$. 
\end{proposition}

\section{PSPACE-Hardness} \label{sec:lowerbound}

In this section, we prove our PSPACE-hardness result.

\pspacehard*

\subsection{Extending Stochastic SAT Hardness}\label{sec:extend-SSAT}

We first extend hardness of approximation for MAX-SSAT instances as in \Cref{thm:k-max-ssat-hard} to instances which in addition satisfy that \emph{deterministic} variables appear in at most $k$ clauses. 

\begin{lemma}\label{lem:ssatm}
There exist constants $k \in \NN$ and $\alpha \in (0,1)$ so that it is PSPACE-hard to compute an $\alpha$-approximation to $\opton(\phi)$ for a MAX-SSAT instance $\phi$ satisfying
\begin{enumerate}[(1)]
    \item no random variable appears negated in any clause of $\phi$, and \label{item:rand-non-negated}
    \item \emph{each} variable (both random and deterministic) appears in at most $k$ clauses of $\phi$.\label{item:both-det-and-rand}
\end{enumerate}
\end{lemma}

We give a polynomial-time reduction from $\alpha$-approximating $\opton(\phi)$ for a MAX-SSAT instance $\phi$
as in \Cref{thm:k-max-ssat-hard} to $\alpha'$-approximating $\opton(\phi')$ on a MAX-SSAT instance $\phi'$ satisfying both properties \ref{item:rand-non-negated} and 
\ref{item:both-det-and-rand} for some $k'=O(1)$ and constant $\alpha'\in (0,1)$.

\paragraph{The reduction.}
For odd (deterministic) $i$, if the variable $x_i$ appears in $a(i)$ clauses in $\phi$, we replace the $j$\textsuperscript{th} occurrence of $x_i$ with a new variable $x_{i, j}$ for $1 \le j \le a(i)$. Let $\phi'$ denote the new 3CNF formula after these replacements. 
We also add clauses to force the optimal online algorithm to set all of $(x_{i, 1}, x_{i, 2}, \ldots x_{i, a(i)})$ equal to each other, without increasing their number of occurrences by more than a constant. Specifically, 
for each odd $i$, we construct via \Cref{lem:explicitexpanders} an expander graph $G_i$ on $a(i)$ vertices with maximum degree at most $d = O(1)$ and expansion at least 1. Associate the vertices of $G_i$ with the literals $(x_{i,1}, x_{i,2}, \ldots, x_{i,a(i)})$ arbitrarily. For any edge in $G_i$ between $x_{i, j}$ and $x_{i, j'}$, add the following two clauses to $\phi'$:
\begin{align}\label{additional-clauses}
(x_{i,j} \vee \overline{x_{i, j'}}) \wedge (\overline{x_{i,j}} \vee x_{i, j'}).
\end{align}
Note that if $x_{i, j} \neq x_{i, j'}$, we satisfy exactly one of these two clauses, while if $x_{i, j} = x_{i, j'}$ we satisfy both. 
The order of variables $x_{i,j}$ and $x_i$ in $\phi$ is some arbitrary order such that variables in $\phi'$ corresponding to (copies of) variables $x_i$ and $x_j$ in $\phi$ appear in an order consistent with the variables $x_i$ and $x_j$ in $\phi$. By adding dummy random variables, we further guarantee that copies of deterministic/random variables in $\phi$ are likewise deterministic/random in $\phi'$.

The following lemma relates the maximum expected number of satisfiable clauses in $\phi$ and $\phi'$, needed to complete our reduction's analysis.
\begin{lemma}\label{opt-phi-opt-phi'}
Let $E_n := \sum_{ \text{odd } i \le n } 2|E(G_{i})|$. Then, the MAX-SSAT instances $\phi$ and $\phi'$ satisfy
\begin{equation*}
\opton(\phi') = \opton(\phi) + E_n.
\end{equation*}
\end{lemma}

\begin{proof}
We first prove $\opton(\phi')\geq \opton(\phi)+E_n$.
Consider an online algorithm $\mathcal{A}$ which for odd $i$ sets $x_{i,1} = x_{i,2} = \ldots = x_{i, a(i)} = b_i$, where $b_i$ is the assignment for $x_i$ of $\opton$ on $\phi$ given the induced history. This algorithm for $\phi'$ is clearly implementable.
Moreover, this algorithm satisfies each of the $E_n$ clauses of form \eqref{additional-clauses}, and satisfies $\opton(\phi)$ of the original clauses in expectation. Hence $\opton(\phi')\geq \mathcal{A}(\phi') = \opton(\phi) + E_n.$ 

We now prove that $\opton(\phi') \le \opton(\phi) + E_n.$
Assume that for some odd $i$, and some fixed history for all variables before $(x_{i,1}, \ldots, x_{i,a(i)})$, an SSAT algorithm $\mathcal{A}$ sets $(x_{i,1}, x_{i,2}, \ldots, x_{i, a(i)})$ such that they do not all take the same value (with some positive probability). Consider the minimum size subset $S \subseteq \{1, 2, \ldots, a(i)\}$ such that flipping all $\{x_{i, j}\}_{j \in S}$ would result in all variables being set to the same value (so, $1 \le |S| \le a(i)/2$). Since the expansion of $G_i$ is at least 1, we know that $|E(S, V \setminus S)| \ge |S|$; flipping all the $\{x_{i, j}\}_{j \in S}$ would hence let us satisfy at least $|S|$ additional clauses of the form \eqref{additional-clauses}, and possibly satisfy $|S|$ fewer clauses corresponding to clauses in $\phi$ containing $x_i$. Thus, $\mathcal{A}$ would satisfy at least as many clauses in expectation by flipping the sign of $\{x_{i,j} \}_{j\in S}$. 
Repeatedly applying this transformation results in an improved online algorithm $\calA'$ as stated in the previous paragraph, from which we find that $\opton$ satisfies at most $\opton(\phi) \leq \calA'(\phi') \leq \opton(\phi') + E_n$ clauses in expectation. The lemma follows.
\end{proof}

We now show that $E_n$ is bounded from above by a constant times $\opton(\phi)$. 
\begin{observation}\label{Enupperboundobs}
$E_n\leq 12d\cdot \opton(\phi).$
\end{observation}
\begin{proof}
Since for each odd $i$, the expander graph $G_i$ contains at most $d$ edges per each of the $a(i)$ occurrences of $i$ in $\phi$, we have that $E_n = \sum_{\textrm{odd }i\leq n} 2|E(G_i)|\leq \sum_{\textrm{odd }i\leq n} 2d\cdot a(i)$.
Next, for $m$ the number of clauses in $\phi$, since $\phi$ is a 3-CNF formula,   $\sum_{\text{odd } i \le n} a(i) \le 3m$. 
Finally, we note that, since 
setting each variable randomly satisfies at least half of the clauses in expectation, 
$ m/2\leq \opton(\phi)$. Combining these observations, we find that
\begin{align}
E_n &= \sum_{\text{odd } i \le n} 2|E(G_i)| \le \sum_{\text{odd } i \le n} d \cdot a(i) \le 6dm \le 12d \cdot \opton(\phi).\qedhere \nonumber
\end{align}
\end{proof}

Given the above, we are now ready to prove \Cref{lem:ssatm}.
\begin{proof}[Proof of \Cref{lem:ssatm}]
Let $\alpha\in (0,1)$ and $k$ be the constants in the statement of \Cref{thm:k-max-ssat-hard}. 
Let $\phi$ be a MAX-SSAT instance as in the statement of that lemma and $\phi'$ be the obtained instance from the reduction of this section, which is polynomial-time, by \Cref{lem:explicitexpanders}.
By construction, no random variable appears negated in any clause, and each variable appears in at most $k'=\max(d+2,k) = O(1)$ clauses.
By \Cref{opt-phi-opt-phi'},  $\opton(\phi') = \opton(\phi) + E_n$. 
Next, we let $Q = \opton(\phi)$, $Q' = E_n$, and $\beta = 12d$, and note that  
$Q'/Q \le \beta$, by 
\Cref{Enupperboundobs}.
Thus, by  \Cref{lem:approxlem}, for  the constant $\alpha' := \left( \frac{\alpha + 12d}{1 + 12d} \right)\in (0,1)$, an $\alpha'$-approximation to $\opton(\phi') = \opton(\phi) + E_n = Q+Q'$ yields an $\alpha$-approximation of $Q=\opton(\phi)$, which is PSPACE-hard, by \Cref{thm:k-max-ssat-hard}. 
\end{proof}

\subsection{Hardness of Algorithms for $\ride$}\label{sec:pspace-reduction}

We are now ready to prove our main theorem about the hardness of $\ride$. Throughout this proof, we will let $k = O(1)$ be the constant in the statement of \Cref{lem:ssatm}.
Denote the variables in an SSAT instance $\phi$ as in \Cref{lem:ssatm} by $(x_1, x_2, \ldots, x_n)$ and the number of clauses of $\phi$ by $m$. Without loss of generality, suppose $n$ is even. 
From $\phi$, we construct a $\ride$ instance $\calI_{\phi}$, with weights $w_{i,t} = w_t$ for each pair $(i,t)\in E$, where we refer to $w_t$ as the weight of ball $t$.
The instance has $2n$ bins, corresponding to the literals $\{x_i, \overline{x_i} \mid i\in [n]\}.$
The instance $\calI_{\phi}$ has $n+m$ balls; we will refer to the first $n$ balls as ``literal balls'' and the final $m$ balls as ``clause balls'' (for reasons that will become clear shortly). For odd $t\leq n$, ball $t$ arrives with probability 1, has weight $1$, and has an edge only to bins $x_{t}$ and $\overline{x_{t}}$. For even $t\leq n$, ball $t$ arrives with probability $1/2$, has weight $1$, and has an edge only to bin $x_{t}$. The last $m$ clause balls $t=n+1,\dots,n+m$ each have weight $\frac{m^4}{2k}$ and arrive with probability $m^{-4}$. The clause ball $t=n+r$ corresponding to clause $C_r$ neighbors only the bins corresponding to literals in $C_r$. (See \Cref{fig:binreduction}.)

\begin{figure}[h] 
    \centering
    \includegraphics{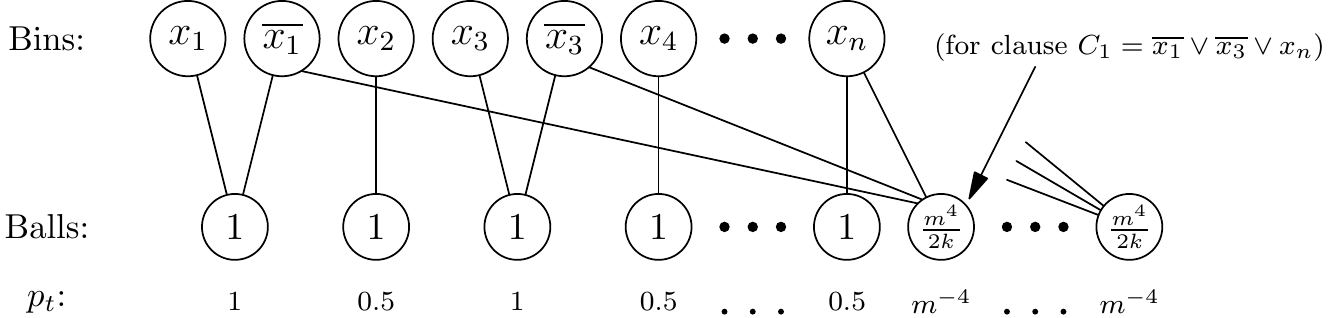}
    \caption{The $\ride$ instance $\calI_{\phi}$}
    \vspace{0.2cm}
     \footnotesize{Bins are labeled by their corresponding literal, while balls are labeled by their weight.}
    \label{fig:binreduction}
\end{figure}

We shall see that $\opton(\calI_\phi)$ and $\opton(\phi)$ are, up to a negligible error term, related by a simple linear relation.
In particular, we will show that 

\begin{align} \label{relation-phi-calIphi}
    \opton(\calI_{\phi}) = 0.75n + \frac{ (1 - m^{-4})^{m-1} }{2k}\cdot \opton(\phi)+ o(1). 
\end{align}

We prove \Cref{relation-phi-calIphi} in the following two lemmas. The first proves that $\opton$ run on $\calI_\phi$ matches all arriving literal balls.

\begin{lemma}\label{match-literal-balls}
    Algorithm $\opton$ matches all arriving literal balls of $\calI_\phi$.
\end{lemma}
\begin{proof}
    Suppose that there is some history $h$ (occurring with probability $q > 0$) after which $\opton$ does not match a literal ball $t$ which arrives; let $\mathcal{A}'$ be the algorithm that follows exactly what $\opton$ does, with the exception that it will match $t$ if $t$ arrives after the history $h$. Then, $$\mathcal{A}'(\calI_{\phi}) - \opton(\calI_{\phi}) \ge q \cdot \left( 1 - k \cdot \frac{m^4}{2k} \cdot m^{-4} \right) = q/2 > 0.$$ Indeed, if the history $h$ occurs, $\mathcal{A}'$ gets a guaranteed profit of 1 from matching $t$ that $\opton$ does not receive. The expected profit $\opton$ gets from having the additional bin available to be potentially matched to clause balls is at most $k \cdot \frac{m^4}{2k} \cdot m^{-4}$, since each literal bin has at most $k$ clause balls adjacent to it, each of which has value $\frac{m^4}{2k}$ and arrives with probability $m^{-4}$. 
    As the above would imply $\calA(\calI_\phi) > \opton(\calI_\phi)$, 
    we conclude that $\opton$ must match each literal ball that arrives. 
\end{proof}

A simple corollary of the above is that $\opton$ gets value of $0.75n$ in expectation from the literal balls it matches.
Moreover, this lemma gives a natural correspondence between $\opton$ on $\calI_\phi$ and algorithms for $\phi$.
The following lemma relies on \Cref{match-literal-balls} to bound the value $\opton$ obtains from the clause balls in terms of the expected number of clauses of $\phi$ satisfied by $\opton$.

\begin{lemma}\label{value-clause-balls}
    Let $B$ be the gain of $\opton$ from clause balls of $\calI_\phi$.
    Then, for some $\delta\in [0,2m^{-1}]$,
    $$\EE[B] = \frac{ (1 - m^{-4})^{m-1} }{2k}\cdot \opton(\phi) + \delta.$$
\end{lemma}
\begin{proof}
By \Cref{match-literal-balls}, $\opton$ matches each arriving literal ball. 
We consider the following natural mapping between MAX-SSAT algorithms $\calA$ on $\phi$ and families of algorithms $\calF_\calA$ which match each literal ball in  $\calI_\phi$.
For odd $t\leq n$, an algorithm $\calA'\in \calF_\calA$ 
matches ball $t$ to bin $\overline{x_t}$ ($x_t$) iff 
algorithm $\calA$ sets $x_t$ to \texttt{True} (\texttt{False}). For even $t\leq n$, if ball $t$ arrives, an algorithm $\calA'\in \calF_\calA$ matches ball $t$ to bin $x_t$; this corresponds to nature setting $x_t = \texttt{False}$. Otherwise, bin $x_t$ is unmatched up to time $m+1$, and we will think of this as nature setting $x_t = \texttt{True}$. (Note that ball $t$ arrives with probability 50\%, so the variables are set to \texttt{True}/\texttt{False} with the correct probability.)
Finally, algorithms $\calA'\in \calF_\calA$ match each arriving clause ball to some available neighboring bin when possible. A simple exchange argument shows that $\opton(\calI_\phi)\in \calF_\calA$ for some algorithm $\calA$.

Let $C$ be the number of clause balls of $\calI_\phi$ that arrive.
Then, with probability $\Pr[C=1]=m \cdot m^{-4} \cdot (1 - m^{-4})^{m-1}$, exactly one such clause ball arrives, equally likely to correspond to any of the $m$ clauses in $\phi$.
On the other hand, a literal $x_t$ (respectively, $\overline{x_t}$) is unmatched by $\calA'\in \calF_\calA$ immediately prior to time $m+1$ iff $\calA(\phi)$ or nature set $x_t$ to \texttt{True} (respectively, \texttt{False}).
We conclude that Algorithm $\calA'\in \calF_{\calA}$ gains $\frac{\calA(\phi)}{m}\cdot \frac{m^4}{2^k}$ expected value from conditioned on a single clause ball arriving.
Thus, the expected gain $\EE[B]$ of $\opton(\calI_\phi)$ from clause balls is at least 
\begin{align}\label{clause-gain-lb}
\EE[B]\geq \EE[B \mid C=1]\cdot \Pr[C=1] =  \frac{(1 - m^{-4})^{m-1}}{2k}\cdot \opton(\phi).
\end{align}

Let $\calA$ be the MAX-SSAT algorithm for which $\opton(\calI_\phi)\in \calF_\calA$.
By the above argument yielding \Cref{clause-gain-lb}, the expected gain of $\opton(\calI_\phi)$ from clause balls conditioned on $C=1$ is precisely
\begin{align}\label{gain-single-clause-ball}
\Pr[B \mid C=1] = \frac{\calA(\phi)}{m} \cdot \frac{m^4}{2k} \le \frac{\opton(\phi)}{m} \cdot \frac{m^4}{2k}.
\end{align}

Next, we note that the probability that multiple clause balls arrive is inverse polynomial in $m$.  \begin{align}\label{prob-two-balls-or-more}
\PP[C \ge 2] = \sum_{t=2}^m \binom{m}{t} m^{-4t} (1 - m^{-4})^{m-t} \le \sum_{t=2}^m m^t \cdot  m^{-4t} \le m^{-6} + m \cdot m^{-9} \leq 2m^{-6}.
\end{align}
On the other hand, conditioned on at multiple clause balls arriving, the expected profit of $\opton$ from clause balls is at most 
\begin{align}\label{opton-given-two-balls-or-more}
\EE[B \mid C\geq 2 ] \le m\cdot  \frac{m^4}{2k} \le m^5.
\end{align}

Combining equations \eqref{gain-single-clause-ball}, \eqref{prob-two-balls-or-more} and \eqref{opton-given-two-balls-or-more}, we find that the expected gain of $\opton(\calI_\phi)$ from matching clause balls is at most
\begin{align*}
\EE[B] & = \EE[B \mid C=1]\cdot \Pr[C=1] + \EE[B \mid C\geq 2]\cdot \Pr[C\geq 2] \\
& \leq \frac{\opton(\phi)}{m} \cdot \frac{m^4}{2k}\cdot m\cdot m^{-4} (1 - m^{-4})^{m-1} + m^5\cdot 2m^{-6} \\
& = \frac{(1 - m^{-4})^{m-1}}{2k}\cdot \opton(\phi) +  2m^{-1}.\qedhere
\end{align*}
\end{proof}



We now conclude the reduction, and obtain the proof of our hardness result. 

\begin{proof}[Proof of \Cref{thm:ridehailpspacehard}]
Let $\alpha\in (0,1)$ be the constant from the statement of \Cref{lem:ssatm} and $\phi$ be a MAX-SSAT instance as in the statement of that lemma. Without loss of generality, we assume that $\phi$ has no pairs of consecutive variables $x_{2k-1}$ and $x_{2k}$ which appear in no clauses. (Else, we remove these variable pairs and relabel the remaining variables while preserving parity of indices. This does not change the clauses, nor does it change the expected number of clauses satisfied by $\opton$.)
Next, let $\calI_\phi$ be the obtained $\ride$ instance from the (clearly polynomial-time) reduction of this section; note furthermore than $\calI_\phi$ has all weights and inverse arrival probabilities bounded above by some polynomial in the size of the input. From \Cref{match-literal-balls}, the expected gain of $\opton(\calI_\phi)$ from literal balls is $0.75n$. Combining this with \Cref{value-clause-balls} we find that for $\gamma :=  \frac{ (1 - m^{-4})^{m-1} }{2k}$ and some $\delta\in [0,2m^{-1}]$,
$$\opton(\calI_\phi) = 0.75n +  \gamma\cdot \opton(\phi)+\delta.$$
Next, since $\phi$ is a 3-CNF formula with at least half its variables appear in at least one clause, the number of variables is at most $n\leq 6m$. Moreover, since setting all variables randomly satisfies at least half of the clauses in expectation, we have $m/2\leq \opton(\phi)$. Combining these two observations, we get 
\begin{equation}\label{eqn:0.75natmost8}
0.75n < n\le 12 \cdot \opton(\phi),
\end{equation}
Next, let 
$Q = \gamma \cdot \opton(\phi) + \delta$, $Q' = 0.75n$, and $\beta = \frac{12}{\gamma}$. Note that $Q'/Q \le \beta$ by \Cref{eqn:0.75natmost8}, and that $\beta = O(1)$, since $k=O(1)$.
Therefore, by  \Cref{lem:approxlem}, for  the constant $\alpha' := ( \frac{\alpha\cdot (\gamma+2m^{-1})/\gamma + \beta}{1 + \beta})$, which is in the range $(0,1)$ for sufficiently large $m$, an $\alpha'$-approximation to $\opton(\calI_\phi) = \opton(\phi) + 0.75n = Q+Q'$ yields an $\alpha\cdot (\gamma+2m^{-1})/\gamma$-approximation of $Q \in [\gamma\cdot \opton(\phi), (\gamma+2m^{-1}) \cdot \opton(\phi)]$. By scaling appropriately, this yields an $\alpha$-approximation to $\opton(\phi)$, which is PSPACE-hard to obtain, by \Cref{lem:ssatm}. The theorem follows.
\end{proof}

\section{Algorithmic Results}\label{sec:algo}
In this section we give an algorithm to approximate the profit of $\opton$, for any joint distributions over edge weights of each ball $t$.

\thmalg*

\paragraph{\underline{An LP Relaxation.}}

Our starting point is a linear program (LP) called LP-Match, which we show upper bounds the gain of any online algorithm for $\ride$. Below, the variables we optimize over are $\{x_{i,t}\}$, which we think of as ``the probability that the online algorithm matches ball $t$ to bin $i$''. Recall that ball $t$ arrives with probability $p_t$.

\begin{align}
\textbf{LP-Match:} \hspace{2em} \qquad \max \enspace \enspace & \sum_{i, t}  w_{i,t} \cdot  x_{i,t}  &&&  \nonumber    \\ 
   \text{s.t.} \enspace \enspace
    & \sum_{t} x_{i,t}     \le 1 & & \text{ for all } i \label{eqn:lpatmost1}
  \\ 
  & \sum_{i} x_{i,t} \le p_{t} & & \text{ for all } t 
 \label{eqn:lpatmostqi} 
 \\ 
& x_{i,t} \hfill  \le p_{t} \cdot \left( 1 - \sum_{t' < t} x_{i,t'} \right) & & \text{ for all } i,t
 \label{eqn:lpthird}
 \\ 
 & 
 x_{i,t}
\ge 0 & & \text{ for all } i,t
\end{align}

Denoting by LP-Match($\calI)$ the optimal value of LP-Match on Instance $\calI$, we have the following.

\begin{lemma}\label{LP-bound}
For any $\ride$ instance $\mathcal{I}$, we have that $$\textrm{\emph{LP-Match}}(\calI)\geq \opton(\calI).$$
\end{lemma}

The above lemma is implied by \cite{torrico2017dynamic}.
For completeness, we provide a proof of this lemma below.

\begin{proof}
Let $x^*_{i,t}$ denote the probability that $\opton$ matches bin $i$ to ball $t$. We note that $x^*$ constitutes a feasible solution for LP-Match because (i) the probability $\opton$ matches a bin $i$ is at most 1, (ii) the probability $\opton$ matches a ball $t$
is at most $p_{t}$ (the probability that $t$ arrives), 
(iii) the probability $\opton$ matches a bin $i$ to a ball $t$ is at most $p_{t}$ (the probability $t$ arrives) times $ 1 - \sum_{t' < t} x_{i,t'}  $ (the probability that $i$ is not matched by time $t$),\footnote{This uses the fact that arrival of $t$ is independent of the online algorithm's previous choices. Note that this constraint is not valid for the probabilities induced by an \emph{offline} algorithm, so our LP does not upper bound $OPT_{off}(\calI)$.} and (iv) these probabilities are non-negative. On the other hand, for this $x^*_{i,t}$, the objective of LP-Match is precisely the expected profit of $\opton$ on this instance, and therefore $\textrm{LP-Match}(\calI)\geq \opton(\calI)$.
\end{proof}

\subsection{The Algorithm}

Given a solution to LP-Match, whose objective upper bounds $\opton$ by \Cref{LP-bound}, a natural approach to approximate $\opton$ is to round this solution online. 
By simple ``integrality gap'' examples (see \Cref{sec:LP-Match-obs}), this is impossible to do perfectly. 
Instead, we show how to do so approximately, by rounding a solution to LP-Match while only incurring a $\nicefrac{1}{2}+c$ multiplicative loss in the rounding, for the constant $c := \finalCconstant$. 

For notational simplicity, assume without loss of generality that an optimal solution to LP-Match to the input instance $\calI$ satisfies all Constraints \eqref{eqn:lpatmostqi} at equality, i.e., $\sum_i x_{i,t} = p_{t}$ for all balls $t$. This can be guaranteed by adding a dummy bin $i_{t}$ for each ball $t$ with $w_{i,t}=0$, and setting $x_{i_t,t} \leftarrow p_{t} - \sum_{i} x_{i,t}$. These dummy edges do not affect the gain of $\opton$, nor that of the online algorithm.

After computing a solution to LP-Match as above, our algorithm proceeds iteratively as follows. For each time $t$, if ball $t$ arrives, we pick a single bin $i$ with probability $x_{i,t}/p_{t}$, and if this is bin is vacant (unmatched), we match $(i,t)$ with some probability $q_{i,t}$. (We sometimes refer to this as $i$ \emph{accepts} $t$.) If this did not result in $t$ being matched, we repeat the process a second time, but this time we match $t$ to its picked bin $i$, provided $i$ is vacant, and the edges until time $t$ have nearly saturated Constraint \eqref{eqn:lpatmost1} for $i$. See \Cref{alg:propose-twice}.

\begin{algorithm}[H]
	\begin{algorithmic}[1] 
		\medskip
		\State solve LP-Match for $\{ x_{i,t} \}$\label{line:lp-solve}
		\State add dummy neighbor for each $t$ so that $\sum_{i} x_{i,t} = p_{t}$
		\State $\mathcal{M} \leftarrow \emptyset$ 
		\For{all balls $t=1,2,\dots$} 
		\State pick a single bin $i$ with probability $\frac{x_{i,t}}{p_{t}}$ \label{line:pick1}
		\If{$i$ is unmatched in $\mathcal{M}$}
		\With{\textbf{probability} $q_{i,t} := \min \left( 1, \frac{\nicefrac{1}{2} + c}{1 - \sum_{t' < t} x_{i, t'} \cdot (\nicefrac{1}{2} + c)  } \right)$} 
		\State $\mathcal{M} \leftarrow \mathcal{M} \cup \{ (i,t)\}$ \label{line:probacceptfirstpick} \label{line:updateMfirstpick} \label{line:acceptfirstproposal} 
		\EndWith
		\EndIf
		\If{$t$ is still unmatched in $\mathcal{M}$}\label{line:second-pick-start}
		\State pick a single bin $i$ with probability $\frac{x_{i,t}}{p_{t}}$ \label{line:pick2}
		\If{$i$ is unmatched in $\mathcal{M}$ \textbf{and} $\sum_{t' < t} x_{i, t'} > \frac{\nicefrac{1}{2}-c}{\nicefrac{1}{2}+c}$}
		\State $\mathcal{M} \leftarrow \mathcal{M} \cup \{ (i,t)\}$ \label{line:acceptsecondproposal}
		\EndIf
		\EndIf
		\EndFor				\State \textbf{Output} $\mathcal{M}$
	\end{algorithmic}
	\caption{Rounding LP-Match Online}\label{alg:propose-twice}
\end{algorithm}

By Constraint \eqref{eqn:lpatmostqi}, Lines \ref{line:pick1} and \ref{line:pick2} are well-defined. Also, by Constraint \eqref{eqn:lpatmost1}, \Cref{line:probacceptfirstpick} is well-defined since $c < 1/2$.
We also note that the algorithm clearly outputs a matching. 

As we shall show, our \Cref{alg:propose-twice} fares well in comparison to $\opton(\calI)$. In particular, we will show the following per-edge guarantees.

\begin{theorem}\label{edge-marginals}
    Each edge $(i,t)\in E$ is matched by \Cref{alg:propose-twice} with probability at least
    $$\PP[(i,t)\in \mathcal{M}] \geq x_{i,t}\cdot (\nicefrac{1}{2} + c).$$
\end{theorem}

 \Cref{edge-marginals} implies that our algorithm is a polynomial-time $0.51$-approximation of the optimal online algorithm, thus proving \Cref{thm:alg}.
%
\begin{proof}[Proof of \Cref{thm:alg}]
    All steps of \Cref{alg:propose-twice}, including solving the polynomially-sized LP in \Cref{line:lp-solve}, can be implemented in polynomial time. 
    The approximation ratio follows directly from linearity of expectation, together with \Cref{LP-bound} and \Cref{edge-marginals}.
\end{proof}

The remainder of this section is dedicated to proving \Cref{edge-marginals}. To this end, we consider two events for edge $(i,t)$ being matched---depending on whether it was matched as a first pick or second pick, in \Cref{line:acceptfirstproposal} or \Cref{line:acceptsecondproposal}, respectively.
We bound the probability of an edge being matched either as a first pick or as a second pick in the following sections.

\subsection{Analysis: First Pick}\label{sec:first-pick} 

In this section we bound the probability of an edge being matched as a first pick. That is, the probability that edge $(i,t)$ is added to $\calM$ in \Cref{line:acceptfirstproposal}. We start with the following useful definition.

\begin{definition}
Ball $t$ is \emph{early} for bin $i$ if $\sum_{t'<t} x_{i,t'} \leq \frac{\nicefrac{1}{2}-c}{\nicefrac{1}{2}+c}$. Otherwise, it is \emph{late}. Edge $(i,t)$ is early (late) if $t$ is early (late) for $i$. We use $E_i$ and $L_i$ to denote the early and late balls for $i$, respectively.

\end{definition}

Intuitively, a ball is late for bin $i$ if most balls $t'$ (weighted by $x_{i,t'}$-value) precede $t$. Note that the early/late distinction determines whether or not the probability $q_{i,t}$ in \Cref{line:probacceptfirstpick} is 1. In particular, this probability is less than 1 only if $(i,t)$ is early, and equal to 1 when $(i,t)$ is late. We will use this observation frequently in the subsequent analysis.

For every $(i,t)$, we let $V_{i,t}$ be an indicator random variable for the event that bin $i$ is \emph{vacant} (i.e., unmatched) at time $t$. We additionally let $\mathcal{M}_1 \subseteq \mathcal{M}$ denote the edges in $\mathcal{M}$ added as a result of a bin $i$ accepting a ball's first pick (i.e., in Line \ref{line:acceptfirstproposal}), and $\mathcal{M}_2 \subseteq \mathcal{M}$ denote the edges in $\mathcal{M}$ added as a result of a bin $i$ accepting a ball's second pick (i.e., in Line \ref{line:acceptsecondproposal}). Note that  $\mathcal{M} = \mathcal{M}_1 \sqcup \mathcal{M}_2$. 

The next lemma bounds the probability of an edge $(i,t)$ being matched as a first pick (in \Cref{line:acceptfirstproposal}).

\begin{lemma}\label{first-pick-marginals}
If edge $(i,t)\in E$ is early, then
    $$\PP[(i,t)\in \mathcal{M}_1] =  x_{i,t}\cdot (\nicefrac{1}{2}+c).$$
In addition, for any edge $(i,t)\in E$,
    $$x_{i,t}\cdot (\nicefrac{1}{2}-3c)\leq \PP[(i,t)\in \mathcal{M}_1] \leq x_{i,t}\cdot (\nicefrac{1}{2}+c).$$
    
\end{lemma}

\begin{proof}
Fix $i$. We prove by strong induction that these bounds hold for all edges $(i,t')$ with $t' < t$. The base case, for $t=1$, is vacuously true. Assume the claim holds for all $t' < t$; we will prove it holds for $t$ as well. 

The event $(i,t) \in \mathcal{M}_1$ requires that ball $t$ arrives and bin $i$ is picked in Line \ref{line:pick1}, that bin $i$ is vacant at time $t$, and that bin $i$ accepts the offer. Note that $i$ being vacant at time $t$ is independent from the arrival of $t$, and the first pick of $t$. Therefore,
\begin{equation}\label{eqn:probitinM1}
\PP[(i,t) \in \mathcal{M}_1] = x_{i,t} \cdot \PP[V_{i,t}] \cdot q_{i,t}.
\end{equation}
For this reason, we turn our attention to bounding the probability of $i$ being vacant at time $t$, 
\begin{equation}\label{eqn:probVit}
\PP[V_{i,t}] = 1 - \sum_{t' < t} \PP[(i,t') \in \mathcal{M}] = 1 - \sum_{t' < t} \PP[(i,t') \in \mathcal{M}_1] - \sum_{t' < t} \PP[(i,t') \in \mathcal{M}_2].
\end{equation} 
First, the inductive hypothesis and the definition of $x_{i,t}$ imply the following upper bound on $\PP[V_{i,t}]$.
\begin{align}\label{Prvit-UB}
    \PP[V_{i,t}] \le 1 - \sum_{t' < t, \atop t' \in E_i} \PP[(i,t') \in \mathcal{M}_1] = 1 - \sum_{t' < t, \atop t' \in E_i} x_{i,t'} \cdot (\nicefrac{1}{2} + c).
\end{align}
If $(i,t)$ is early, this bound is tight because $(i,t')$ is early for any $t' < t$; hence, for early $(i,t)$ we have that $\PP[V_{i,t}] = 1 - \sum_{t' < t} x_{i,t'} \cdot (\nicefrac{1}{2} + c)$. Recalling that $q_{i,t} = \frac{\nicefrac{1}{2} + c}{1 - \sum_{t' < t} x_{i, t'} \cdot (\nicefrac{1}{2} + c)  }$ for early $(i,t)$, \Cref{eqn:probitinM1} then implies that $\PP[(i,t) \in \mathcal{M}_1] = x_{i,t} \cdot (\nicefrac{1}{2} + c)$ for early edges $(i,t)$. 

If $(i,t)$ is late, then $\sum_{t' < t, t' \in E_i} x_{i,t'} = \sum_{t' \in E_i} x_{i,t'} \ge \frac{\nicefrac{1}{2} - c}{\nicefrac{1}{2} + c}$. Hence, by \Cref{Prvit-UB} we have that 
\begin{align}\label{vacancy-ub}
\PP[V_{i,t}] \le 1 - \left( \frac{\nicefrac{1}{2} - c}{\nicefrac{1}{2} + c} \right) \cdot (\nicefrac{1}{2} + c) = \nicefrac{1}{2} + c.
\end{align}
Again, \Cref{eqn:probitinM1} then implies that $\PP[(i,t) \in \mathcal{M}_1] \le x_{i,t} \cdot (\nicefrac{1}{2} + c)$ for late edges $(i,t)$. 

Finally, we lower bound $\PP[(i,t) \in \mathcal{M}_1]$ for late $(i,t)$. To do so, we lower bound $\PP[V_{i,t}]$; here, our analysis must account for the fact that late edges can be matched in either $\mathcal{M}_1$ or $\mathcal{M}_2$. Hence, we first note that for any $t'<t$ that is late for $i$, we have, similarly to \Cref{Prvit-UB} that the probability of edge $(i,t')$ being matched as a second pick is at most 
\begin{equation}\label{eqn:it'inM2upperbound}
\Pr[(i,t')\in \mathcal{M}_2] \leq x_{i,t'} \cdot \Pr[V_{i,t'}] \leq x_{i,t'}\cdot ( \nicefrac{1}{2} + c).
\end{equation} 
Now, combining equations \eqref{eqn:probVit} and \eqref{eqn:it'inM2upperbound}, we lower bound $\PP[V_{i,t}]$ as follows:
\begin{align*}
    \PP[V_{i,t}] 
    &\ge  1 - \sum_{t' < t} x_{i,t'} \cdot (\nicefrac{1}{2} + c) - \sum_{t' < t, \atop t' \in L_i} x_{i,t'}\cdot(\nicefrac{1}{2}+c)
    \ge 1 - (\nicefrac{1}{2} + c) - \left( 1 - \ctwo \right)\cdot  (\nicefrac{1}{2} + c),
    \end{align*}
which simplifies to 
\begin{align}\label{vacancy-lb}
\PP[V_{i,t}] \ge \nicefrac{1}{2} - 3c.
\end{align}
Again, \Cref{eqn:probitinM1} then implies that $\Pr[(i,t)\in \calM_1]\geq x_{i,t}\cdot (\nicefrac{1}{2}-3c)$.
\end{proof}

The proof of \Cref{first-pick-marginals} yields upper and lower bounds on $\PP[V_{i,t}]$ (equations \eqref{Prvit-UB}, \eqref{vacancy-ub} and \eqref{vacancy-lb}), which will prove useful later. For convenience, we extract these bounds in the following corollary.

\begin{corollary}\label{VitBounds}
For any edge $(i,t)$, we have that $\PP[V_{i,t}] \ge \nicefrac{1}{2} - 3c$. For any late $(i,t)$, we have that $\PP[V_{i,t}] \le \nicefrac{1}{2} +c$.
For any early $(i,t)$, we have that $\PP[V_{i,t}] = 1 - \sum_{t' < t} x_{i,t'} \cdot (\nicefrac{1}{2} + c).$
\end{corollary}

Given \Cref{first-pick-marginals}, in order to prove \Cref{edge-marginals}, we wish to prove that the second attempt of $t$ to match will ensure late edges $(i,t)$ a probability of at least $x_{i,t}\cdot 4c$ of being matched. This is the meat of our analysis, and the next section is dedicated to its proof.

\subsection{Analysis: Second Pick}\label{sec:second-pick} 

In this section we prove that the second pick of ball $t$, in Lines \ref{line:second-pick-start}-\ref{line:acceptsecondproposal}, does indeed increase the probability of late edges $(i,t)$ to be matched.
In particular, we prove the following theorem.

\begin{theorem}\label{second-pick-marginals}
    For any late edge $(i,t)\in E$,
    $$\PP[(i,t)\in \mathcal{M}_2] \geq  x_{i,t}\cdot 4c.$$
\end{theorem}

Before proving the above theorem, we provide some useful intuition and outline the challenges the proof of \Cref{second-pick-marginals} needs to overcome.

By \Cref{first-pick-marginals}, the probability of a late edge $(i,t)$ being matched as a first pick is at least 
\begin{align}\label{prob-it-first-pick}
    \Pr[(i,t)\in \mathcal{M}_1] \geq x_{i,t}\cdot (\nicefrac{1}{2}-3c).
\end{align}
Moreover, by the same lemma, each edge $(i,t)\in E$ (whether early or late) is matched as a first pick with probability at most $\Pr[(i,t)\in \mathcal{M}_1] \leq x_{i,t}\cdot (\nicefrac{1}{2}+c)$. 
Denote by $A_t$ the event that $t$ arrives and denote by $U_1(t)$ the event that $t$ is unmatched after its first pick of $i_1=j$. Then, we have $$ \Pr[U_1(t) \mid A_t, i_1 = j] = 1 - \PP[V_{j, t}] \cdot q_{j,t}.$$ 
If $(j,t)$ is late, then because $\PP[V_{j,t}] \le \nicefrac{1}{2} + c$ by \Cref{VitBounds}, the above quantity is at least $\nicefrac{1}{2} - c$. If $(j,t)$ is early, then because $\PP[V_{j,t}] = 1 - \sum_{t' < t} x_{j,t'} (\nicefrac{1}{2} + c)$, by \Cref{VitBounds}, combined with the definition of $q_{j,t}$, we have that the above quantity is exactly equal to $\nicefrac{1}{2} - c$. 
In summary,
\begin{align}\label{prob-t-rejected}
    \PP[U_1(t) \mid A_t, i_1 = j] \ge \nicefrac{1}{2}-c.
\end{align}

Now, we recall that for late edges $(i,t)$, we have that $q_{i,t} = 1$. So, a late edge $(i,t)$ is matched iff $i$ is vacant by time $t$ and $i$ is picked in \Cref{line:pick1} or \Cref{line:pick2}.
One might then be tempted to guess that 
$\Pr[(i,t)\in \mathcal{M}_2 \mid U_1(t)]$ is equal to $\Pr[(i,t)\in \mathcal{M}_1]$, which by \eqref{prob-it-first-pick} and \eqref{prob-t-rejected} would imply that
$\Pr[(i,t)\in \mathcal{M}_2] \geq  x_{i,t} \cdot(\nicefrac{1}{2}-c)\cdot(\nicefrac{1}{2}-3c)\geq x_{i,t}\cdot 4c$ (the last inequality using $c\leq 0.01$), as desired.
\subsubsection{The Key Challenges}

There are two key issues with the simplistic argument above. 

\paragraph{\underline{Challenge 1: Re-drawing $i$.}}
Unfortunately, conditioning on $U_1(t)$ does not result in the probability of $(i,t)$ being matched in the second pick equalling that of it being matched in the first pick. To see this, suppose a ball $t$ was late for a single bin $i$, and $x_{i,t}/p_{t}=1$. In that case, conditioning on $U_1(t)$ is equivalent to conditioning on $i$ being occupied (matched) before time $t$. Consequently, for this late edge $(i,t)$, we have that $\Pr[(i,t)\in \mathcal{M}_1] \geq x_{i,t}\cdot (\nicefrac{1}{2}-3c)$ by \Cref{first-pick-marginals}, while 
$\Pr[(i,t)\in \mathcal{M}_2 \mid U_1(t)] = 0$, which implies that the second pick does not increase the probability of $(i,t)$ to be matched \emph{at all}, as $\Pr[(i,t)\in \mathcal{M}_2] = 0$(!).

This is where Constraint \eqref{eqn:lpthird} of LP-Match comes in: This constraint implies that if $t$ is late for bin $i$, then the probability that $i$ was picked in \Cref{line:pick1} at time $t$ conditioned on arrival of $t$ is at most 
$$\frac{x_{i,t}}{p_{t}} \leq 1-\sum_{t'<t} x_{i,t'} \leq 1-\frac{\nicefrac{1}{2}-c}{\nicefrac{1}{2}+c} = \frac{2c}{\nicefrac{1}{2}+c} \leq 4c.$$ 
This implies that there is a (high) constant probability of $i$ not being picked in \Cref{line:pick1}.

\begin{lemma}\label{i1!=i}
For any late edge $(i,t)$, for $i_1$ the bin picked in \Cref{line:pick1} at time $t$, 
$$\Pr[i_1 \neq i \mid A_t] \geq 1-4c.$$
\end{lemma}

\paragraph{\underline{Challenge 2: Positive Correlation Between Bins.}}
\Cref{i1!=i} alone does not resolve our problems.
Suppose that ball $t$ is late for all bins for which $x_{i,t}\neq 0$, and all these bins have perfectly positively correlated matched status, i.e., $V_{i,t} = V_{j,t}$ for all bins $i,j$ always. If this were the case, then we would have that $\Pr[V_{i,t} \mid U_1(t)] = 0$, since if $t$ is not matched to its first $i_1$, then $i_1$ and $i$ must both have been matched before. This again would result in 
$\Pr[(i,t)\in \mathcal{M}_2] = 0$.

To overcome the above, we show that the above scenario does not occur. In particular, we show that while positive correlations between different bins' matched statuses are possible, such correlations cannot be too large. More formally, we show the following.

\begin{lemma}\label{nearly-neg-corr}
    For any time $t$ and bins $i\neq j$, we have that
    $$\Cov(V_{i,t},V_{j,t}) \leq 12c.$$
\end{lemma}

The crux of our analysis is proving \Cref{nearly-neg-corr}. Using it, we will be able to argue that for any late edge $(i,t)$, the probability that $i$ is free at time $t$, conditioned  on $U_1(t)$ and on the first pick satisfying $i_1\neq i$ (a likely event, by \Cref{i1!=i}), is not changed much compared to the unconditional probability of $i$ being free at time $t$. 
In particular, this implies that the probability of $(i,t)$ being matched as a second pick, conditioned on $U_1(t)$, is not too much smaller compared to its probability of being matched as a first pick.
In particular, we will show that $\Pr[(i,t)\in \calM_2]\geq x_{i,t}\cdot 4c$, for sufficiently small $c>0$, as stated in \Cref{second-pick-marginals}.

We prove that lemmas \ref{i1!=i} and \ref{nearly-neg-corr} indeed imply \Cref{second-pick-marginals}, as outlined above, in \Cref{sec:second-pick-wrapup}. But first, we turn to proving our key technical lemma, namely \Cref{nearly-neg-corr}.

\subsubsection{Bounding Correlations of Occupancies}

To bound the correlation of vacancy indicators, it is convenient to define the indicator random variable $O_{i,t} := 1 - V_{i,t}$, which indicate whether $i$ is occupied (i.e., matched) at time $t$. We additionally decompose the variables $O_{i,t}$ into two variables, based on whether $i$ was matched (became occupied) along an early or late edge. In particular, we let $O_{i,t}^E \leq O_{i,t}$ be an indicator for the event that $i$ is matched along an early edge before $t$ arrives. Similarly, we let $O_{i,t}^L := O_{i,t} - O^E_{i,t}$ be an indicator for the event that $i$ is matched along a late edge before $t$ arrives. 
To bound the pairwise correlations of variables $O_{i,t}$, we will show that $O^E_{i,t}$ contributes most of the probability mass of $O_{i,t}$, and that the variables $O_{i,t}^E$ and $O_{j,t}^E$ are negatively correlated. To prove this negative correlation, we will prove the following, stronger statement.


\begin{lemma}\label{requisite-NA-lemma}
For any time $t$, the variables $\{O^E_{i,t}\}_i$ are negatively associated (NA).
\end{lemma}
\begin{proof}
For every edge $(i,t)$, let $X_{i,t}$ be the indicator random variable for the event that ball $t$ arrives and picks bin $i$ as its first pick. 
Let $Y_{i,t} \sim \text{Ber} \left(  q_{i,t} \right)$ be an indicator for the event that bin $i$ accepts, i.e., it will be matched to ball $t$ if it arrives and picks $i$ as its first pick and $i$ is free. 

For fixed $t$, the variables $\{X_{i,t} \}$ are 0/1 random variables whose sum is at most 1 always, so they are NA by the 0-1 Principle (\Cref{0-1-NA}). 
On the other hand, the variables $\{Y_{i,t}\}_{i}$ are independent, and hence NA. Moreover, $\{X_{i,t}\}_{i}$, $\{Y_{i,t}\}_{i}$ are mutually independent distributions, and so by closure of NA under independent union (\Cref{NA-closure}), we also have that $\{X_{i,t}, Y_{i,t}\}_{i}$ is NA. Likewise, the lists $\{X_{i,t}, Y_{i,t}\}_{i}$ are mutually independent as we vary $t$; again using closure of NA under independent union we find that $\{X_{i,t}, Y_{i,t}\}_{i,t}$ are also NA. 

Fix $t$. For each bin $i$, let $t_i$ denote the largest $t' < t$ so that $(i,t')$ is early. We note that bin $i$ cannot be matched as a second pick to any $t' \le t_i$. So, it is matched along an early edge before $t$ arrives if and only if there are some $t' \le t_i$ and $r$ such that ball $t'$ arrives and picks bin $i$, and bin $i$ accepts the proposal (for the smallest such $t'$, bin $i$ is guaranteed to be free). Therefore, we have that
$$O_{i,t}^E = \bigvee_{t' \le t_i} (X_{i,t'} \wedge Y_{i,t'}).$$

Note that we have written $\{O_{i,t}^E\}_i$ as the output of monotone non-decreasing functions defined on disjoint subsets of the variables in $\{X_{i,t}, Y_{i,t}\}_{i,t}$. Hence, by closure of NA under monotone function composition (\Cref{NA-closure}), we have that $\{O_{i,t}^E\}_i$ are NA. 
\end{proof}

By \Cref{NA:neg-corr}, the above lemma implies that any $O^E_{i,t}$ and $O^E_{j,t}$ are negatively correlated.
\begin{corollary}\label{ML-neg-corr}
For any time $t$ and bins $i\neq j$, we have that
    $\Cov(O^E_{i,t}, O^E_{j,t})\leq 0.$
\end{corollary}


We are now ready to prove \Cref{nearly-neg-corr}. 

\begin{proof}[Proof of \Cref{nearly-neg-corr}]
First, we show that the probability of a bin $i$ being matched along a late edge before time $t$ is small, which we later use to bound the covariance of $O^L_{i,t}$ and other binary variables. 
First, if $(i,t)$ is not late, then trivially, 
$\Pr[O^L_{i,t}]=0\leq 4c.$
Otherwise, we have that $\sum_{t'<t: (i,t') \textrm{ early }} x_{i,t'} \geq \frac{\nicefrac{1}{2}-c}{\nicefrac{1}{2}+c}$.
Thus, by \Cref{first-pick-marginals}, we have that $\PP[O_{i,t}^E] > \frac{\nicefrac{1}{2} - c}{\nicefrac{1}{2} + c} \cdot ( \nicefrac{1}{2} + c) = \nicefrac{1}{2} - c$. On the other hand, by \Cref{VitBounds}, we also have that $\PP[V_{i,t}] \ge \nicefrac{1}{2} - 3c$. Therefore, we find that here, too, the probability of $O^L_{i,t}$ is small.
\begin{align*}
\Pr[O_{i,t}^L] &= \Pr[O_{i,t}] - \Pr[O_{i,t}^E] < (\nicefrac{1}{2} + 3c) - (\nicefrac{1}{2} - c) = 4c.
\end{align*}
From the above, we find that regardless of whether or not $(i,t)$ is late, we have that
\begin{align}\label{eqn:OitL}
\Pr[O_{i,t}^L]\leq 4c.
\end{align}
Therefore, using the additive law of covariance for $\Cov(O_{i,t}, O_{j,t}) = \Cov(1 - O_{i,t}, 1 - O_{j,t}) = \Cov(V_{i,t}, V_{j,t}) $, we obtain the desired bound, 
    \begin{align*}
    \Cov(V_{i,t}, V_{j,t}) &= 
    \Cov(O^E_{i,t} + O^L_{i,t}, O^E_{j,t} + O^L_{j,t}) \\
        & = \Cov(O^E_{i,t}, O^E_{j,t}) + \Cov(O^E_{i,t}, O^L_{j,t}) + \Cov(O^L_{i,t}, O^E_{j,t}) + \Cov(O^L_{i,t}, O^L_{j,t}) \\
        & \leq 0 + \Pr[O^E_{i,t}, O^L_{j,t}] + \Pr[O^L_{i,t}, O^E_{j,t}] + \Pr[O^L_{i,t}, O^L_{j,t}] & \textrm{Cor. \ref{ML-neg-corr}} \\
        & \leq 0 + \Pr[O^L_{j,t}] + \Pr[O^L_{i,t}] + \Pr[O^L_{i,t}]  \\
        &\le 12c. & \textrm{Eq. \eqref{eqn:OitL}} & \qedhere \end{align*}
 \end{proof}

\subsubsection{Putting it All Together}\label{sec:second-pick-wrapup}

We are now ready to use weak positive correlation (if any) between vacancy indicators $V_{i,t}$ and $V_{j,t}$. 
In particular, we will show that the probability of bin $i$ to be occupied a time $t$ is not changed much when conditioning on $A_t$ (arrival of $t$), the first picked bin at time $t$ being $i_1\neq i$, and $U_1(t)$ (ball $t$ bot being matched to its first pick). 

\begin{lemma} \label{Oitconditionedonrejected} For any late edge $(i,t)$, we have that
$$\Pr[ O_{i,t} \mid A_t, i_1 \neq i, U_1(t)] \le \PP[ O_{i, t}] \cdot \left( 1 + \frac{12c}{(\nicefrac{1}{2} - c)^2} \right).$$ 
\end{lemma}

\begin{proof}
To analyze the conditional probability above, we first look at $\PP[O_{i,t}, A_t, i_1 = j, U_1(t)]$. This is the probability of bin $i$ being occupied at time $t$, ball $t$ arriving and picking $j$ as its first pick, and not being matched due to this first pick. Note that $A_t$ and the first pick is independent of bins' occupancy statuses at time $t$. Additionally, we notice that with probability $1 - q_{j,t}$ bin $j$ will deterministically reject. With probability $q_{j,t}$, it rejects if and only if $j$ is occupied. So, for any $j \neq i$,
\begin{equation}\label{eqn:Oitjoint}
\PP[O_{i,t}, A_t, i_1 = j, U_1(t)] = \PP[O_{i,t}] \cdot \PP[A_t, i_1 = j] \cdot  \left( (1 - q_{j,t})  + q_{j,t} \cdot \PP[O_{j,t} \mid O_{i,t}] \right).
\end{equation}

We now turn to relating the last term in the above product, namely $(1 - q_{j,t})  + q_{j,t} \cdot \PP[O_{j,t} \mid O_{i,t}]$, to its "unconditional" counterpart, $\PP[U_1(t) \mid A_t , i_1 =j] = (1 - q_{j,t}) + q_{j,t} \cdot \PP[O_{j,t}]$. For notational convenience, we  which we abbreviate by $$z_{i,j,t} := (1 - q_{j,t})  + q_{j,t} \cdot \PP[O_{j,t} \mid O_{i,t}].$$ Recalling that $\Cov(O_{i,t}, O_{j,t}) = \Cov(V_{i,t}, V_{j,t})\leq 12c$, by \Cref{nearly-neg-corr}, we have
\begin{equation}\label{eqn:Oitjt}
\PP[O_{j,t} \mid O_{i,t}] =  \frac{\PP[O_{j,t}, O_{i,t}]}{\PP[O_{i,t}]} = \frac{\PP[O_{j,t}] \cdot \PP[O_{i,t}] + \text{Cov}(O_{j,t}, O_{i,t}) }{\PP[O_{i,t}]} \le  \PP[O_{j,t}] + \frac{12c}{\PP[O_{i,t}]}.
\end{equation} 
Hence,
\begin{align}
\nonumber z_{i,j,t} &\le (1 - q_{j,t})  + q_{j,t} \cdot \left( \PP[O_{j,t}] + \frac{12c}{\PP[O_{i,t}]} \right) & & \text{(Eq. (\ref{eqn:Oitjt}))}\\
\nonumber  &\le (1 - q_{j,t})  + q_{j,t} \cdot \left( \PP[O_{j,t}] + \frac{12c}{\nicefrac{1}{2}-c} \right) && \text{(Cor. \ref{VitBounds}, } c<\nicefrac{1}{2}) \\
\nonumber &= \PP[U_1(t) \mid A_t, i_1 = j] + q_{j,t} \cdot \frac{12c}{\nicefrac{1}{2} - c} \\
&\le \PP[U_1(t) \mid A_t, i_1 = j] \cdot \left( 1 + \frac{12c}{(\nicefrac{1}{2} - c)^2} \right) &&\text{(Eq. (\ref{prob-t-rejected}), } q_{j,t} \le 1)
\end{align} 

Using this bound in \Cref{eqn:Oitjoint} and summing over all $j \neq i$, we have $$\PP[O_{i,t}, A_t, i_1 \neq i, U_1(t)] \le \PP[O_{i,t}] \cdot \PP[A_t, i_1 \neq i, U_1(t)] \cdot \left( 1 + \frac{12c}{(\nicefrac{1}{2} - c)^2} \right).$$ The desired inequality therefore follows by Bayes' theorem.
\end{proof}

With this lemma in place, we are ready to conclude this section by proving \Cref{second-pick-marginals}, i.e. that $\PP[(i,t) \in \mathcal{M}_2] \ge x_{i,t} \cdot 4c$ for any late edge $(i,t)$. 

\begin{proof}[Proof of \Cref{second-pick-marginals}]
We start by bounding 
\begin{equation}\label{eqn:itinM2firstbound}
\PP[(i,t) \in \mathcal{M}_2] \ge  \PP[(i,t) \in \mathcal{M}_2 \mid A_t, i_1 \neq i, U_1(t)] \cdot \PP[A_t, i_1 \neq i, U_1(t)].
\end{equation} 
In words, the probability $(i,t)$ is matched as a second pick is at least the probability of the same event and $i_1\neq i$. By \Cref{i1!=i} we know that $\PP[A_t, i_1 \neq i] \ge p_{t}\cdot (1 - 4c)$; by \Cref{prob-t-rejected}, we know that $ \Pr[U_1(t) \mid A_t, i_1 = j] \ge \nicefrac{1}{2} - c$ for any $j \neq i$. As a consequence, by Bayes' theorem and our choice of $c < \nicefrac{1}{4}$, we have that
\begin{equation}\label{eqn:13}
\PP[A_t, i_1 \neq i, U_1(t)] = \PP[A_t, i_1 \neq i] \cdot \PP[ U_1(t) \mid A_t, i_1 \neq i] \ge p_{t} \cdot (1-4c) \cdot (\nicefrac{1}{2} - c).
\end{equation}
Next, we note that
\begin{equation}\label{eqn14}
\PP[(i,t) \in \mathcal{M}_2 \mid A_t, i_1 \neq i, U_1(t)] = \frac{x_{i,t}}{p_{t}} \cdot \PP[V_{i,t} \mid A_t, i_1 \neq i, U_1(t)]
\end{equation}
because conditioned on $A_t$, picking someone other than $i$ first, and being rejected, we will match $(i,t)$ exactly when $t$'s second pick is $i$ and $i$ is vacant. 

\Cref{Oitconditionedonrejected} yields the following lower bound on the probability of $[V_{i,t} \mid A_t$, $i_1 \neq i, U_1(t)]$:
\begin{align} 
\nonumber \PP[V_{i,t} \mid A_t, i_1 \neq i, U_1(t)] &= 1 - \PP[O_{i,t} \mid A_t, i_1 \neq i, U_1(t)] \\
\nonumber &\ge 1 - \PP[O_{i,t}] \cdot \left( 1 + \frac{12c}{(\nicefrac{1}{2} - c)^2} \right) \\
\nonumber &= \PP[V_{i,t}] - \frac{12c}{(\nicefrac{1}{2} - c)^2} \cdot (1 - \PP[V_{i,t}] ) \\
&\ge \nicefrac{1}{2} - 3c  - \frac{12c}{(\nicefrac{1}{2} - c)^2} \cdot (\nicefrac{1}{2} + 3c)  &(\text{Cor. \ref{VitBounds}}) \label{eqn15}
\end{align} 
Combining equations \ref{eqn14} and \ref{eqn15} we thus have 
\begin{equation}\label{eqn:16}
\PP[(i,t) \in \mathcal{M}_2 \mid A_t, i_1 \neq i, U_1(t)] \ge \frac{x_{i,t}}{p_{t}} \cdot \left( \nicefrac{1}{2} - 3c  - \frac{12c}{(\nicefrac{1}{2} - c)^2} \cdot (\nicefrac{1}{2} + 3c) \right).
\end{equation} 
Putting it all together, equations (\ref{eqn:itinM2firstbound}), (\ref{eqn:13}), and (\ref{eqn:16}) and our choice of (sufficiently small) $c=0.01$ imply the desired inequality,
\begin{align*}
\PP[(i,t) \in \mathcal{M}_2] & \ge  \frac{x_{i,t}}{p_{t}} \cdot  \left( \nicefrac{1}{2} - 3c  - \frac{12c}{(\nicefrac{1}{2} - c)^2} \cdot (\nicefrac{1}{2} + 3c) \right) \cdot p_{t} \cdot (1 - 4c) \cdot (\nicefrac{1}{2} - c) \geq x_{i,t}\cdot 4c. \qedhere
\end{align*}
\end{proof}

\section{Generalizing the Algorithm}\label{sec:general-algo}
Our algorithm and its analysis of \Cref{sec:algo} generalize seamlessly to a setting in which weights of each online node $t$ are drawn from discrete joint distributions. For brevity, we only outline the small changes in the LP, algorithm and analysis here.

\paragraph{Problem Statement.} 
We are given a complete bipartite graph, with vertices of one side (bins) give up front, and vertices of the other side (balls) arriving sequentially, with ball $t$ arriving at time $t$ (with probability one).
The vector of edge weights of any ball $t$, denoted by  $w^t:=(w_{1,t},w_{2,t},\dots)$, is drawn from some discrete joint distribution, $w^t \sim \calD_t$. 
The vector of all edge weights, $w:=(w^1,w^2,\dots)$, is drawn from the product distribution, $w\sim \calD := \prod_t \calD_t$. That is, the weights of any ball's edges may be arbitrarily correlated, but weights of different balls' edges are independent.
We assume that these discrete distributions are given explicitly, e.g., via a list of tuples of the form $(v_{t,j}, p_{t,j})$ with $p_{t,j} := \Pr_{\calD_t}[w^t=v_{t,j}]$. 
We note that the problem considered in previous sections is a special instance of this problem with each $\calD_t$ consisting of two-point distributions, with one of the possible realizations of $w^t\sim \calD_t$ being the all-zeros vector.

\paragraph{Generalizing LP-Match.}
The generalization of LP-Match now has decision variables $y_{i,t,j}$, which we think of as proxies for the probability of edge $(i,t)$ being matched by the optimal online algorithm when ball $t$'s edge weights are $w^t=v_{t,j}$.
Generalizing the argument behind Constraint \eqref{eqn:lpthird}, we note that $w^t$ is independent of bin $i$ not being matched by the optimal online algorithm by time $t$. From this we obtain Constraint \eqref{LP-gen:newconstraint} below. The remaining constraints of the obtained LP (below) are matching constraints.

\begin{figure}[h]
\begin{align}
\textbf{LP-Match-Gen:} \hspace{2em} \qquad \max \enspace \enspace & \sum_{i, t, j}  w_{i,t,j} \cdot  y_{i,t,j}  &&&  \nonumber    \\ 
   \text{s.t.} \enspace \enspace
    & \sum_{t}\sum_{j} y_{i,t,j}     \le 1 & & \text{ for all } i \nonumber 
  \\ 
  & \sum_{i} y_{i,t,j} \le p_{t,j} & & \text{ for all } t,j 
 \nonumber
 \\ 
& y_{i,t,j} \hfill  \le p_{t,j} \cdot \left( 1 - \sum_{t' < t} \sum_{j'} y_{i,t',j'} \right) & & \text{ for all } i,t,j  
 \label{LP-gen:newconstraint}
 \\ 
 & 
 y_{i,t,j}     \ge 0 & & \text{ for all } i,t,j\nonumber
\end{align}
\end{figure}

\paragraph{Generalizing the algorithm.}
Our general algorithm will match each edge $(i,j)$ when $w^t=v_{t,j}$ with marginal probability at least probability 
\begin{align*}
    \Pr[(i,t)\in \calM, w^t=v_{t,j}] \geq y_{i,t,j}\cdot (\nicefrac{1}{2}+c).
\end{align*}
To do so, when ball $t$ arrives, we first observe the realization of the edge weight vector $w^t=v_{t,j}$. Then, When picking a bin $i$ (either as first or second pick) at time $t$, we now do so with probability $\frac{y_{i,t,r}}{p_{t,j}}$.
Moreover, we take $q_{i,t}:=\min\left(1,\frac{\nicefrac{1}{2}+c}{1-\sum_{t'<t}\sum_{j'} y_{i,t',j'}\cdot (\nicefrac{1}{2}+c)}\right)$ to be the probability of a vacant picked bin $i$ to be matched to ball $t$ by the algorithm.
The dummy nodes $i_t$ are now assigned values $y_{i_t,t,j}\leftarrow p_{t,j}-\sum_{i} y_{i,t,j}$ for each $j$.
Apart from this, the algorithm is unchanged.
We note that this algorithm can be implemented in polynomial time in the size of the input (the representation of $\calD$).

\paragraph{Generalizing the Analysis.}
Extending the analysis of \Cref{alg:propose-twice} to this more general problem is a rather simple syntactic generalization. We therefore only outline the changes in the analysis.
Broadly, all changes needed for the analysis require us to refine our claims as follows. Denote by $R_t$ a random variable denoting the random index of the weight vector of edges of $t$. That is, $R_t=j \iff w^t=v_{t,j}$. Then, all our bounds for the probability of $(i,t)$ being matched (as a first or second pick, or either) now need to refer to $R_t=j$, and relate to $y_{i,t,j}$. So, for example, \Cref{first-pick-marginals} will be restated to show that for each early edge $(i,t)$ and index $j$, we have that $\Pr[(i,t)\in \calM_1, R_t=j] = y_{i,t,r}\cdot (\nicefrac{1}{2}+c)$, and for any edge $(i,t)$, we have that $y_{i,t,r}\cdot (\nicefrac{1}{2}-3c)\leq \Pr[(i,t)\in \calM_1, R_t=j] \leq y_{i,t,r}\cdot (\nicefrac{1}{2}+c)$. 
\Cref{requisite-NA-lemma} requires some care in setting up the NA variables to prove that $O^E_{i,t}$ are NA, by also accounting for the realization of $R_t$, with indicators $[R_t=j]$, which are NA by the 0-1 Principle (\Cref{0-1-NA}).
Apart from that, the proofs are essentially unchanged, except for replacing occurrences of $A_t$ by $R_t=j$ in every probability conditioned on arrival of $t$, and appropriately replacing $\frac{x_{i,t}}{p_t}$ by $\frac{y_{i,t,j}}{p_{t,j}}$.

\section{Conclusions and Open Questions}\label{sec:conclusion}

We studied the online stochastic max-weight bipartite matching problem through the lens of approximation algorithms, rather than that of competitive analysis. 
In particular, we study the efficient approximability of the optimal online algorithm on any given input. 
On the one hand, we show that the optimal online algorithm cannot be approximated beyond some constant (barring shocking developments in complexity theory). On the other hand, we present a polynomial-time online algorithm which yields a $0.51$ approximation of the optimal online algorithm's gain---surpassing the approximability threshold of $\nicefrac{1}{2}$ of the optimal offline algorithm. 
Many intriguing research questions remain.

First, it is natural to further study the efficient approximability of our problem. We suspect that much better approximation guarantees are achievable; in particular, \cite{torrico2017dynamic} suggests a family of additional constraints strengthening our LP relaxation, possibly leading to improved approximation.
One might also ask if our general algorithmic approach can be extended to \emph{implicitly} represented weight distribution $\calD$.
For example, what can one show if $\calD_t$ is itself a product distribution, $\calD_t = \prod_i \calD_{i,t}$, with $w_{i,t}\sim \calD_{i,t}$?
A related interesting question is to obtain better approximation for the widely-studied special case of balls drawn from some i.i.d distribution (see, e.g., \cite{manshadi2012online,haeupler2011online,huang2018online,mahdian2011online,karande2011online}).

More broadly, one might ask how well one can approximate the optimal online algorithm of online Bayesian selection problems under the numerous  constraints studied in the literature, including matroid and matroid intersections, knapsack constraints, etc.
For which of these problems is the online optimum easy to compute? Which admit a PTAS? Which admit constant approximations? Which are hard to approximate? We are hopeful that the ideas developed here, both algorithmic, as well as our new hardness gadgets, will prove useful when exploring this promising research agenda.

\paragraph{Acknowledgements.} We thank the anonymous EC'21 reviewers and Neel Patel for useful comments which helped improve the presentation of this manuscript, and we thank the authors of \cite{torrico2017dynamic} for drawing our attention to their work.

\appendix
\section{Hardness of Computing  Approximately-Optimal Online Policies}\label{app:policyhardnesschernoff}

In this section we justify our claim that a hardness result for approximating the \emph{value} achieved by the optimal online algorithm implies a hardness result for the computation of the \emph{decisions} made by an (approximately) optimal online algorithm. 
Let $\alpha$ be as in \Cref{thm:ridehailpspacehard}. 

\begin{claim}
No polynomial-time algorithm computes the decisions made of an online algorithm which $\left( \frac{\alpha + 1}{2} \right)$-approximates the optimal online \ride algorithm, unless $PSPACE=BPP$.\footnote{BPP denotes the decision problems solvable in polynomial times by \emph{randomized} algorithms which fail with probability at most $\nicefrac{1}{3}$.}
\end{claim}

\begin{proof}

We reduce from the problem of computing an $\alpha$-approximation to the profit obtained by $\opton$ for a fixed input $\mathcal{I}$, with polynomially bounded weights and inverse arrival probabilities. Let $\textsc{OPT}$ denote this profit. Let $P$ denote the maximum possible profit for $\mathcal{I}$ for any realization of the randomness.

Assume we could compute the decisions made by an algorithm $\mathcal{A}$ which achieves an  $\left( \frac{\alpha + 1}{2} \right)
$-approximation. For some parameter $T$, use these decisions to run the algorithm on $T$ independent instantiations of a given input and record the profits as $X_1$, $X_2$, $\ldots$, $X_T$. Let $\bar{X}$ denote the sample average $\bar{X} := \frac{1}{T} \sum_{i=1}^T X_i$. Using the Chernoff-Hoeffding bound, we can bound the probability $\bar{X}$ deviates from its expectation as
$$\PP \left[ \Big| \bar{X} - \EE[\bar{X}] \Big| \ge \left( \frac{1 - \alpha}{4} \cdot \textsc{OPT} \right) \right] \le 2\cdot \exp \left( - \frac{2T^2 ( \frac{1-\alpha}{4} \cdot \textsc{OPT} )^2}{T \cdot P^2} \right) \le \exp \left(- \Theta \left( T \cdot \frac{\textsc{OPT}^2}{P^2} \right) \right).$$
Take $T = \Theta(n \cdot P^2/\textsc{OPT}^2)$; note this is polynomial in the size of the input, as long as all weights and inverse arrival probabilities of $\mathcal{I}$ are polynomially bounded. Then, 
$$\PP \left[ \Big| \bar{X} - \EE[\bar{X}] \Big| \ge \left( \frac{1 - \alpha}{4} \cdot \textsc{OPT} \right) \right] \le \exp \left(- \Theta \left( n \right) \right).$$ so
we can clearly in polynomial time compute a $\bar{X}$ that is, w.h.p., at most $\left( \frac{1-\alpha}{4} \right) \cdot \textsc{OPT}$ far away from a $\left( \frac{\alpha + 1}{2} \right)$-approximation to $\textsc{OPT}$. In particular, $$\bar{X} \in \left[ \textsc{OPT} \left( \frac{3 \alpha + 1}{4} \right), \textsc{OPT} \left( \frac{5 - \alpha}{4} \right) \right].$$ We immediately observe that the quantity $\bar{X} - \textsc{OPT}\left(\frac{1 - \alpha}{4}\right)$ is hence in the interval $[ \textsc{OPT} \cdot \alpha, \textsc{OPT}]$.  Hence, w.h.p., we have given an $\alpha$-approximation to $\textsc{OPT}$. As we demonstrated this problem to be PSPACE-complete, if we can do this in polynomial time w.h.p. then $\text{PSPACE} = \text{BPP}$. 
\end{proof}

\section{Omitted Proofs of \Cref{sec:prelims}}\label{app:ommittedproofs}

In this section we provide proofs deferred from \Cref{sec:prelims}, restated below for ease of reference.

\approxfact*
\begin{proof}
As $f(x)=\frac{\alpha+x}{1+x} = 1-\frac{1-\alpha}{1+x}$ is monotone increasing in $x\geq -1$ for $\alpha\in (0,1)$, we have that 
$(\frac{\alpha + \beta}{1 + \beta}) \geq \frac{\alpha + Q'/Q}{1+Q'/Q} = \frac{\alpha\cdot Q+Q'}{Q+Q'}$. Thus, 
An $(\frac{\alpha + \beta}{1 + \beta})$-approximation to $Q + Q'$ yields a number $T$ in the range $$\Big[ \frac{\alpha + \beta}{1 + \beta} \cdot \left( Q + Q' \right) , Q + Q' \Big] \subseteq [\alpha\cdot Q + Q', Q+Q'].$$ 
Subtracting $Q'$ from $T$ then yields a number $T-Q'$ in the range $[\alpha \cdot Q, Q]$.
\end{proof}

Next, we provide a proof of the underlying PSPACE-hardness result of \citet{condon1997random} used in our reductions.
\randomdebaters*
\begin{proof}
This lemma follows from the proof in \cite{condon1997random}; here, we briefly explain why. 

In that paper, the authors prove their main result that $\text{RPCD}(\log n, 1) = \text{PSPACE}$ in Theorem 2.4. Using this theorem, they prove that it is PSPACE-hard to approximate MAX-SSAT in Theorem 3.1. In their proof, they start with a language $L$ in PSPACE and an input $x$, and construct an RPCDS for $L$ flipping $O(\log n)$ coins and reading $O(1)$ bits of the debate. From this, they construct a MAX-SSAT instance $\phi$ such that if $x \in L$, all clauses of $\phi$ can be satisfied with probability 1, while if $x \notin L$ there is no way to satisfy more than an $\alpha < 1$ fraction of the clauses of $\phi$. Their construction of $\phi$ builds a constant-size 3CNF for each possible realization of the $O(\log n)$ coin flips, and takes the conjunction of these 3CNFs. Each constant-size 3CNF has variables corresponding to the bits of the debate that $V$ queries for a specific realization of the coin-flips. Hence, to show that $\phi$ only has each random variable appear in $O(1)$ clauses, it suffices to show that each random-bit in the RPCDS constructed is queried for only $O(1)$ realizations of the coin flips. 

To show this, we turn to the construction of the RPCDS used to prove Theorem 2.4. Via Lemma 2.1, the authors first show that it is sufficient to consider RPCDSs where the verifier can read a constant number of \emph{rounds} of Player 1 (and not just a constant number of bits). 

In Lemma 2.3, the authors describe their protocol for a verifier $V$ which can read $O(1)$ rounds of Player 1. Note that the random coins in this protocol are used to select a ``random odd-numbered round $k > 1$" and a ``random bit of round $k-1$ of Player 0." In fact, this is the \emph{only} time that the verifier reads a random bit of Player 0. So, in this construction, each random bit is only queried in $O(1)$ realizations of the coin flips. With Lemma 2.1, the authors transform this RPCDS to one that only reads a constant number of bits. We note that this transformation only impacts the strings that player 1 writes, and does not affect the coin flips or the bits of player $0$ read. 

From this, it holds that the MAX-SSAT instance $\phi$ constructed in Theorem 3.1 has each random variable appear in $O(1)$ clauses. That instance does not yet satisfy the property that random variables only appear non-negated. Condon et al. give a fix for this in the proof of Theorem 3.3; we briefly note that after the modification provided in this proof, it will still hold that random variables appear in $O(1)$ clauses. 
\end{proof}
\section{LP-Match: Additional Observations}\label{sec:LP-Match-obs}

Here we make a few additional observations concerning the usefulness of Constraint \eqref{eqn:lpthird} and LP-Match in general, as well as some natural limits to this LP.

First, we note that LP-Match captures the optimal online algorithm \emph{precisely} for the classic single-item prophet inequality problem. That is, for \ride instances with a single bin $i$, solutions to this LP can be rounded online losslessly.

\begin{observation}
    LP-Match$(\calI)= \opton(\calI)$ for any \ride instance $\calI$ with a single bin $i$. 
\end{observation}
\begin{proof}
    Consider the following online algorithm, which starts by computing a solution $\vec{y}$ to LP-Match. Next, upon arrival of ball $t$ with with $w_{i,t}=w_{i,t,r}$ (i.e., $R_t=r$), match $(i,t)$ with probability 
    \begin{align*}
    \frac{y_{i,t,r}}{p_{t,r}\cdot \left(1-\sum_{t'<t}\sum_{r'} y_{i,t',r'}\right)}.
    \end{align*}
    This last quantity is indeed a probability, by Constraint \ref{eqn:lpthird}.
    A simple proof by induction shows that for each $t$ and $r$, we have that $\Pr[(i,t)\in \calM, R_t=r] = y_{i,t,r}$, and consequently $\Pr[F_{i,t}]=1-\sum_{t'<t}\sum_{r'} y_{i,t',r'}$, from which we obtain the inductive step, as
    \begin{align*}
        \Pr[(i,t)\in \calM, R_t=r] & = p_{t,r}\cdot \frac{y_{i,t,r}}{p_{t,r}\cdot \left(1-\sum_{t'<t}{r'} y_{i,t',r'}\right)}\cdot \left(1-\sum_{t'<t}\sum_{r'} y_{i,t',r'}\right) = y_{i,t,r}.
    \end{align*}
    By linearity of expectation, this online algorithm for instance $\calI$ has expected reward precisely
    \begin{align*}
    \sum_{i,t,r} w_{i,t,r}\cdot y_{i,t,r}= \textrm{LP-Match}(\calI).
    \end{align*}
    Consequently, $\opton(\calI)\geq \textrm{LP-Match}(\calI)$. The opposite inequality follows from \Cref{LP-bound}.
\end{proof}

On the other hand, for general \ride instances, there is a limit to the approximation guarantees obtainable using LP-Match.
In particular, simple examples show that there is a gap between the upper bound given by LP-Match and the expected profit of $\opton$, appropriately restricting the approximation guarantees provable using this LP. This is to be expected, given our work in \Cref{sec:lowerbound}. We present a simple example of such a gap instance below.

\begin{observation}
    There exists a $\ride$ instance $\calI$ with $w_{i,t}\in \{0,1\}$ for all $(i,t)\in E$ for which  LP-Match$(\calI)\geq \nicefrac{8}{7}\cdot \opton(\calI)$. 
\end{observation}
\begin{proof}
    We consider an instance $\calI$ with three balls and two bins. For $k=1,2$, ball $t=k$ has with probability $p_{k,0}=1/2$ edge weights $w_{i,t}=0$ for all $i$. With the remaining probability $p_{k,1}=1/2$, its edges have weights $w_{k,k}=1$ and $w_{k,3-k}=0$.
    The last ball has weights $w_{3,k}=1$ for all bins $k=1,2$ with probability one.
    An optimal solution to LP-Match on this Instance $\calI$ assigns $y_{k,k,1}=1/2$ for $k=1,2$, and $y_{3,k,1}=1/2$ for $k=1,2$, achieving an objective value of $\sum_{i,t,r} y_{i,t,r} = 2$.
    However, with probability $\nicefrac{1}{4}$, both of the first two balls have all their edge weights zero, and so an online algorithm can at most achieve an expected value of $7/4$. That is, $\opton(\calI)\leq \nicefrac{7}{8}\cdot \textrm{LP-Match}(\calI)$.
\end{proof}

\section{Unweighted Hardness}
We briefly make the observation that our previous hardness proof also gives a hardness result for \ride instances where all arriving passengers have weight 1. Given the similarity to our previous argument, we only detail the changes that must be made. 

\begin{observation} 
It is PSPACE-hard to approximate the optimal online \ride algorithm within a factor $1-o(1)$, even for \ride instances with binary weights.
\end{observation}

\begin{proof}

We will simply take the construction from \Cref{sec:pspace-reduction} and make all arriving balls have weight 1. In particular, for an SSAT instance $\phi$ as in \Cref{lem:ssatm}, we define the unweighted \ride instance $\mathcal{I}_{\phi}$ as follows in \Cref{fig:binreductionappendix}.

\begin{figure}[h] 
    \centering
    \includegraphics{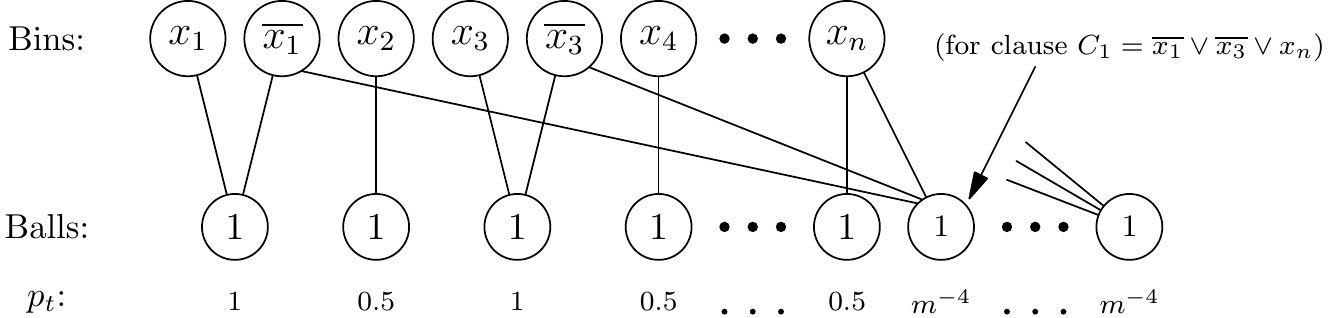}
    \caption{The unweighted $\ride$ instance $\mathcal{I}_{\phi}$}
    \label{fig:binreductionappendix}
    \vspace{0.2cm}
     \footnotesize{Bins are labeled by their corresponding literal, while balls are labeled by their weight.}
\end{figure}

Analogously to \Cref{match-literal-balls}, we can clearly see that $\opton$ matches all arriving literal balls of $\mathcal{I}_{\phi}$, and hence gets an expected profit of at least $0.75n$. Analogously to \Cref{value-clause-balls}, breaking into cases based on the number of arrived balls demonstrates that the expected profit $\opton$ will get from the clause balls is at most $$\opton(\phi) \cdot m^{-4} \cdot (1 - m^{-4})^{m-1} + 2m^{-5} = m^{-4} \left( \opton(\phi) \cdot (1 - m^{-4})^{m-1} + o(1) \right).$$ In summary, the profit of $\opton$ on the instance $\mathcal{I}_{\phi}$ is $$0.75n + m^{-4} \left( \opton(\phi) \cdot (1 - m^{-4})^{m-1} + \delta \right) $$ for some $\delta = o(1)$. 

Apply \Cref{lem:approxlem} with $Q' = 0.75n$ and $Q = m^{-4}(\opton(\phi) \cdot (1-m^{-4})^{m-1} + \delta)$. Note $Q'/Q \le \beta$ for $\beta = O(\text{poly}(n,m))$. Hence an $\left( \frac{\alpha + \beta}{1+ \beta} \right)$-approximation to $Q + Q'$ yields an $\alpha$-approximation to $Q$. Take $\alpha$ to be a sufficiently large constant less than $1$ such that it is PSPACE-hard to obtain an $\alpha$-approximation to $Q$. As $$\frac{\alpha + \beta}{1+ \beta} = 1 - \frac{1-\alpha}{1+\beta}=  1 - \frac{1}{O(\text{poly}(n,m))},$$ it holds it is PSPACE-hard to obtain an approximation to unweighted $\ride$ instances within a factor of $1 - \frac{1}{O(\text{poly}(n,m))}$. 
\end{proof}

\bibliographystyle{plainnat}
\bibliography{abb,a}

\begin{thebibliography}{50}
\providecommand{\natexlab}[1]{#1}
\providecommand{\url}[1]{\texttt{#1}}
\expandafter\ifx\csname urlstyle\endcsname\relax
  \providecommand{\doi}[1]{doi: #1}\else
  \providecommand{\doi}{doi: \begingroup \urlstyle{rm}\Url}\fi

\bibitem[Abolhassani et~al.(2017)Abolhassani, Ehsani, Esfandiari, Hajiaghayi,
  Kleinberg, and Lucier]{abolhassani2017beating}
Melika Abolhassani, Soheil Ehsani, Hossein Esfandiari, MohammadTaghi
  Hajiaghayi, Robert Kleinberg, and Brendan Lucier.
\newblock Beating 1-1/e for ordered prophets.
\newblock In \emph{Proceedings of the 49th Annual ACM Symposium on Theory of
  Computing (STOC)}, pages 61--71, 2017.

\bibitem[Agrawal et~al.(2020)Agrawal, Sethuraman, and
  Zhang]{agrawal2020optimal}
Shipra Agrawal, Jay Sethuraman, and Xingyu Zhang.
\newblock On optimal ordering in the optimal stopping problem.
\newblock In \emph{Proceedings of the 21st ACM Conference on Economics and
  Computation (EC)}, pages 187--188, 2020.

\bibitem[Alaei(2014)]{alaei2014bayesian}
Saeed Alaei.
\newblock Bayesian combinatorial auctions: Expanding single buyer mechanisms to
  many buyers.
\newblock \emph{SIAM Journal on Computing (SICOMP)}, 43\penalty0 (2):\penalty0
  930--972, 2014.

\bibitem[Alaei et~al.(2012)Alaei, Hajiaghayi, and Liaghat]{alaei2012online}
Saeed Alaei, MohammadTaghi Hajiaghayi, and Vahid Liaghat.
\newblock Online prophet-inequality matching with applications to ad
  allocation.
\newblock In \emph{Proceedings of the 13th ACM Conference on Electronic
  Commerce (EC)}, pages 18--35, 2012.

\bibitem[Anari et~al.(2019)Anari, Niazadeh, Saberi, and
  Shameli]{anari2019nearly}
Nima Anari, Rad Niazadeh, Amin Saberi, and Ali Shameli.
\newblock Nearly optimal pricing algorithms for production constrained and
  laminar bayesian selection.
\newblock In \emph{Proceedings of the 20th ACM Conference on Economics and
  Computation (EC)}, pages 91--92, 2019.

\bibitem[Aouad and Sarita{\c{c}}(2020)]{aouad2020dynamic}
Ali Aouad and {\"O}mer Sarita{\c{c}}.
\newblock Dynamic stochastic matching under limited time.
\newblock In \emph{Proceedings of the 21st ACM Conference on Economics and
  Computation (EC)}, pages 789--790, 2020.

\bibitem[Arora et~al.(1998)Arora, Lund, Motwani, Sudan, and
  Szegedy]{arora1998proof}
Sanjeev Arora, Carsten Lund, Rajeev Motwani, Madhu Sudan, and Mario Szegedy.
\newblock Proof verification and the hardness of approximation problems.
\newblock \emph{Journal of the ACM (JACM)}, 45\penalty0 (3):\penalty0 501--555,
  1998.

\bibitem[Asadpour et~al.(2017)Asadpour, Goemans, M{\k{a}}dry, Gharan, and
  Saberi]{asadpour2017log}
Arash Asadpour, Michel~X Goemans, Aleksander M{\k{a}}dry, Shayan~Oveis Gharan,
  and Amin Saberi.
\newblock An ${O}(\log n/\log \log n)$-approximation algorithm for the
  asymmetric traveling salesman problem.
\newblock \emph{Operations Research}, 65\penalty0 (4):\penalty0 1043--1061,
  2017.

\bibitem[Chawla et~al.(2010)Chawla, Hartline, Malec, and
  Sivan]{chawla2010multi}
Shuchi Chawla, Jason~D Hartline, David~L Malec, and Balasubramanian Sivan.
\newblock Multi-parameter mechanism design and sequential posted pricing.
\newblock In \emph{Proceedings of the 42nd Annual ACM Symposium on Theory of
  Computing (STOC)}, pages 311--320, 2010.

\bibitem[Chen et~al.(2016)Chen, Hu, Li, Li, Liu, and Lu]{chen2016combinatorial}
Wei Chen, Wei Hu, Fu~Li, Jian Li, Yu~Liu, and Pinyan Lu.
\newblock Combinatorial multi-armed bandit with general reward functions.
\newblock In \emph{Proceedings of the 30th Annual Conference on Neural
  Information Processing Systems (NIPS)}, pages 1659--1667, 2016.

\bibitem[Condon et~al.(1997)Condon, Feigenbaum, Lund, and
  Shor]{condon1997random}
Anne Condon, Joan Feigenbaum, Carsten Lund, and Peter Shor.
\newblock Random debaters and the hardness of approximating stochastic
  functions.
\newblock \emph{SIAM Journal on Computing (SICOMP)}, 26\penalty0 (2):\penalty0
  369--400, 1997.

\bibitem[Correa et~al.(2017)Correa, Foncea, Hoeksma, Oosterwijk, and
  Vredeveld]{correa2017posted}
Jos{\'e} Correa, Patricio Foncea, Ruben Hoeksma, Tim Oosterwijk, and Tjark
  Vredeveld.
\newblock Posted price mechanisms for a random stream of customers.
\newblock In \emph{Proceedings of the 18th ACM Conference on Economics and
  Computation (EC)}, pages 169--186, 2017.

\bibitem[Correa et~al.(2018)Correa, Foncea, Hoeksma, Oosterwijk, and
  Vredeveld]{correa2019recent}
Jos{\'{e}} Correa, Patricio Foncea, Ruben Hoeksma, Tim Oosterwijk, and Tjark
  Vredeveld.
\newblock Recent developments in prophet inequalities.
\newblock \emph{SIGecom Exchanges}, 17\penalty0 (1):\penalty0 61--70, 2018.

\bibitem[Correa et~al.(2019)Correa, Foncea, Pizarro, and
  Verdugo]{correa2019pricing}
Jos{\'e} Correa, Patricio Foncea, Dana Pizarro, and Victor Verdugo.
\newblock From pricing to prophets, and back!
\newblock \emph{Operations Research Letters}, 47\penalty0 (1):\penalty0 25--29,
  2019.

\bibitem[Dubhashi and Ranjan(1996)]{dubhashi1996balls}
Devdatt Dubhashi and Desh Ranjan.
\newblock Balls and bins: A study in negative dependence.
\newblock \emph{BRICS Report Series}, 3\penalty0 (25), 1996.

\bibitem[D\"utting et~al.(2020)D\"utting, Feldman, Kesselheim, and
  Lucier]{dutting2020prophet}
Paul D\"utting, Michal Feldman, Thomas Kesselheim, and Brendan Lucier.
\newblock Prophet inequalities made easy: Stochastic optimization by pricing
  nonstochastic inputs.
\newblock \emph{SIAM Journal on Computing (SICOMP)}, 49\penalty0 (3):\penalty0
  540--582, 2020.

\bibitem[D{\"u}tting et~al.(2020)D{\"u}tting, Kesselheim, and
  Lucier]{dutting2020log}
Paul D{\"u}tting, Thomas Kesselheim, and Brendan Lucier.
\newblock An ${O}(\log \log m)$ prophet inequality for subadditive
  combinatorial auctions.
\newblock In \emph{2020 IEEE 61st Annual Symposium on Foundations of Computer
  Science (FOCS)}, pages 306--317. IEEE, 2020.

\bibitem[Ezra et~al.(2020)Ezra, Feldman, Gravin, and Tang]{ezra2020online}
Tomer Ezra, Michal Feldman, Nick Gravin, and Zhihao~Gavin Tang.
\newblock Online stochastic max-weight matching: prophet inequality for vertex
  and edge arrival models.
\newblock In \emph{Proceedings of the 21st ACM Conference on Economics and
  Computation (EC)}, pages 769--787, 2020.

\bibitem[Feldman et~al.(2015)Feldman, Gravin, and
  Lucier]{feldman2015combinatorial}
Michal Feldman, Nick Gravin, and Brendan Lucier.
\newblock Combinatorial auctions via posted prices.
\newblock In \emph{Proceedings of the 26th Annual ACM-SIAM Symposium on
  Discrete Algorithms (SODA)}, pages 123--135, 2015.

\bibitem[Feldman et~al.(2016)Feldman, Svensson, and
  Zenklusen]{feldman2016online}
Moran Feldman, Ola Svensson, and Rico Zenklusen.
\newblock Online contention resolution schemes.
\newblock In \emph{Proceedings of the 27th Annual ACM-SIAM Symposium on
  Discrete Algorithms (SODA)}, pages 1014--1033, 2016.

\bibitem[Feng et~al.(2021)Feng, Niazadeh, and Saberi]{feng2021two}
Yiding Feng, Rad Niazadeh, and Amin Saberi.
\newblock Two-stage stochastic matching with application to ride hailing.
\newblock In \emph{Proceedings of the 32nd Annual ACM-SIAM Symposium on
  Discrete Algorithms (SODA)}, pages 2862--2877, 2021.

\bibitem[Fu et~al.(2018)Fu, Li, and Xu]{fu2018ptas}
Hao Fu, Jian Li, and Pan Xu.
\newblock A ptas for a class of stochastic dynamic programs.
\newblock In \emph{Proceedings of the 45th International Colloquium on
  Automata, Languages and Programming (ICALP)}, pages 56:1--56:14, 2018.

\bibitem[Gabber and Galil(1981)]{gabber1981explicit}
Ofer Gabber and Zvi Galil.
\newblock Explicit constructions of linear-sized superconcentrators.
\newblock \emph{Journal of Computer and System Sciences}, 22\penalty0
  (3):\penalty0 407--420, 1981.

\bibitem[Gamlath et~al.(2019)Gamlath, Kapralov, Maggiori, Svensson, and
  Wajc]{gamlath2019online}
Buddhima Gamlath, Michael Kapralov, Andreas Maggiori, Ola Svensson, and David
  Wajc.
\newblock Online matching with general arrivals.
\newblock In \emph{Proceedings of the 60th Symposium on Foundations of Computer
  Science (FOCS)}, pages 26--37, 2019.

\bibitem[Goel et~al.(2010)Goel, Guha, and Munagala]{goel2010probe}
Ashish Goel, Sudipto Guha, and Kamesh Munagala.
\newblock How to probe for an extreme value.
\newblock \emph{ACM Transactions on Algorithms (TALG)}, 7\penalty0
  (1):\penalty0 1--20, 2010.

\bibitem[Gravin and Wang(2019)]{gravin2019prophet}
Nikolai Gravin and Hongao Wang.
\newblock Prophet inequality for bipartite matching: Merits of being simple and
  non adaptive.
\newblock In \emph{Proceedings of the 20th ACM Conference on Economics and
  Computation (EC)}, pages 93--109, 2019.

\bibitem[Haeupler et~al.(2011)Haeupler, Mirrokni, and
  Zadimoghaddam]{haeupler2011online}
Bernhard Haeupler, Vahab~S Mirrokni, and Morteza Zadimoghaddam.
\newblock Online stochastic weighted matching: Improved approximation
  algorithms.
\newblock In \emph{Proceedings of the 7th Conference on Web and Internet
  Economics (WINE)}, pages 170--181. 2011.

\bibitem[Hajiaghayi et~al.(2007)Hajiaghayi, Kleinberg, and
  Sandholm]{hajiaghayi2007automated}
Mohammad~Taghi Hajiaghayi, Robert Kleinberg, and Tuomas Sandholm.
\newblock Automated online mechanism design and prophet inequalities.
\newblock In \emph{Proceedings of the 22nd AAAI Conference on Artificial
  Intelligence (AAAI)}, pages 58--65, 2007.

\bibitem[Hartline(2012)]{hartline2012approximation}
Jason~D Hartline.
\newblock Approximation in mechanism design.
\newblock \emph{American Economic Review}, 102\penalty0 (3):\penalty0 330--36,
  2012.

\bibitem[Hill and Kertz(1992)]{hill1992survey}
Theodore~P Hill and Robert~P Kertz.
\newblock A survey of prophet inequalities in optimal stopping theory.
\newblock \emph{Contemporary Mathematics}, 125:\penalty0 191--207, 1992.

\bibitem[Hill et~al.(1982)Hill, Kertz, et~al.]{hill1982comparisons}
Theodore~P Hill, Robert~P Kertz, et~al.
\newblock Comparisons of stop rule and supremum expectations of iid random
  variables.
\newblock \emph{The Annals of Probability}, 10\penalty0 (2):\penalty0 336--345,
  1982.

\bibitem[Huang et~al.(2018)Huang, Tang, Wu, and Zhang]{huang2018online}
Zhiyi Huang, Zhihao~Gavin Tang, Xiaowei Wu, and Yuhao Zhang.
\newblock Online vertex-weighted bipartite matching: Beating 1-1/e with random
  arrivals.
\newblock In \emph{Proceedings of the 45th International Colloquium on
  Automata, Languages and Programming (ICALP)}, pages 1070--1081, 2018.

\bibitem[Joag-Dev and Proschan(1983)]{joag1983negative}
Kumar Joag-Dev and Frank Proschan.
\newblock Negative association of random variables with applications.
\newblock \emph{The Annals of Statistics}, pages 286--295, 1983.

\bibitem[Karande et~al.(2011)Karande, Mehta, and Tripathi]{karande2011online}
Chinmay Karande, Aranyak Mehta, and Pushkar Tripathi.
\newblock Online bipartite matching with unknown distributions.
\newblock In \emph{Proceedings of the 43rd Annual ACM Symposium on Theory of
  Computing (STOC)}, pages 587--596, 2011.

\bibitem[Karger(2001)]{karger2001randomized}
David~R Karger.
\newblock A randomized fully polynomial time approximation scheme for the
  all-terminal network reliability problem.
\newblock \emph{SIAM review}, 43\penalty0 (3):\penalty0 499--522, 2001.

\bibitem[Kessel et~al.(2021)Kessel, Saberi, Shameli, and
  Wajc]{kessel2021stationary}
Kristen Kessel, Amin Saberi, Ali Shameli, and David Wajc.
\newblock The stationary prophet inequality problem.
\newblock \emph{arXiv preprint arXiv:2107.10516}, 2021.

\bibitem[Khursheed and Lai~Saxena(1981)]{khursheed1981positive}
Alam Khursheed and KM~Lai~Saxena.
\newblock Positive dependence in multivariate distributions.
\newblock \emph{Communications in Statistics - Theory and Methods}, 10\penalty0
  (12):\penalty0 1183--1196, 1981.

\bibitem[Kleinberg and Weinberg(2019)]{kleinberg2019matroid}
Robert Kleinberg and S~Matthew Weinberg.
\newblock Matroid prophet inequalities and applications to multi-dimensional
  mechanism design.
\newblock \emph{Games and Economic Behavior}, 113:\penalty0 97--115, 2019.

\bibitem[Krengel and Sucheston(1978)]{krengel1978semiamarts}
Ulrich Krengel and Louis Sucheston.
\newblock On semiamarts, amarts, and processes with finite value.
\newblock \emph{Probability on Banach spaces}, 4:\penalty0 197--266, 1978.

\bibitem[Lubotzky et~al.(1988)Lubotzky, Phillips, and
  Sarnak]{lubotzky1988ramanujan}
Alexander Lubotzky, Ralph Phillips, and Peter Sarnak.
\newblock Ramanujan graphs.
\newblock \emph{Combinatorica}, 8\penalty0 (3):\penalty0 261--277, 1988.

\bibitem[Lucier(2017)]{lucier2017economic}
Brendan Lucier.
\newblock An economic view of prophet inequalities.
\newblock \emph{ACM SIGecom Exchanges}, 16\penalty0 (1):\penalty0 24--47, 2017.

\bibitem[Mahdian and Yan(2011)]{mahdian2011online}
Mohammad Mahdian and Qiqi Yan.
\newblock Online bipartite matching with random arrivals: an approach based on
  strongly factor-revealing lps.
\newblock In \emph{Proceedings of the 43rd Annual ACM Symposium on Theory of
  Computing (STOC)}, pages 597--606, 2011.

\bibitem[Manshadi et~al.(2012)Manshadi, Gharan, and Saberi]{manshadi2012online}
Vahideh~H Manshadi, Shayan~Oveis Gharan, and Amin Saberi.
\newblock Online stochastic matching: Online actions based on offline
  statistics.
\newblock \emph{Mathematics of Operations Research}, 37\penalty0 (4):\penalty0
  559--573, 2012.

\bibitem[Papadimitriou(1985)]{papadimitriou1985games}
Christos~H Papadimitriou.
\newblock Games against nature.
\newblock \emph{Journal of Computer and System Sciences}, 31\penalty0
  (2):\penalty0 288--301, 1985.

\bibitem[Provan and Ball(1983)]{provan1983complexity}
J~Scott Provan and Michael~O Ball.
\newblock The complexity of counting cuts and of computing the probability that
  a graph is connected.
\newblock \emph{SIAM Journal on Computing}, 12\penalty0 (4):\penalty0 777--788,
  1983.

\bibitem[Saberi and Wajc(2021)]{saberi2021greedy}
Amin Saberi and David Wajc.
\newblock The greedy algorithm is not optimal for online edge coloring.
\newblock In \emph{48th International Colloquium on Automata, Languages, and
  Programming (ICALP 2021)}, pages 109:1--109:18, 2021.

\bibitem[Samuel-Cahn(1984)]{samuel1984comparison}
Ester Samuel-Cahn.
\newblock Comparison of threshold stop rules and maximum for independent
  nonnegative random variables.
\newblock \emph{the Annals of Probability}, 12\penalty0 (4):\penalty0
  1213--1216, 1984.

\bibitem[Segev and Singla(2021)]{segev2020efficient}
Danny Segev and Sahil Singla.
\newblock Efficient approximation schemes for stochastic probing and prophet
  problems.
\newblock In \emph{Proceedings of the 22nd ACM Conference on Economics and
  Computation (EC)}, pages 793--794, 2021.

\bibitem[Torrico and Toriello(2017)]{torrico2017dynamic}
Alfredo Torrico and Alejandro Toriello.
\newblock Dynamic relaxations for online bipartite matching.
\newblock \emph{arXiv preprint arXiv:1709.01557}, 2017.

\bibitem[Valiant(1979)]{valiant1979complexity}
Leslie~G Valiant.
\newblock The complexity of enumeration and reliability problems.
\newblock \emph{SIAM Journal on Computing}, 8\penalty0 (3):\penalty0 410--421,
  1979.

\end{thebibliography}

\end{document}